%% file: main.tex
\documentclass[twoside]{article}

\usepackage[round]{natbib}

\usepackage[utf8]{inputenc} 
\usepackage[T1]{fontenc}    
\usepackage{url}            
\usepackage{booktabs}       
\usepackage{amsfonts}       
\usepackage{nicefrac}       
\usepackage{microtype}      

\usepackage{subcaption}
\usepackage{multirow}

\input{package}

\input{macro}

\usepackage[margin=.8in]{geometry}

\begin{document}
	
	%
	
	%
	
	
	\title{
		Guarantees of Stochastic Greedy Algorithms for 
		Non-monotone Submodular Maximization 
		with Cardinality Constraint
	}
	
	\author{Shinsaku Sakaue
	\\
NTT Communication Science Laboratories}
	
	\maketitle
	
	\begin{abstract}
	Submodular maximization with a cardinality constraint can model various problems, and those problems are often very large in practice. For the case where objective functions are monotone, many fast approximation algorithms have been developed. The stochastic greedy algorithm (SG) is one such algorithm, which is widely used thanks to its simplicity, efficiency, and high empirical performance. However, its approximation guarantee has been proved only for monotone objective functions. When it comes to non-monotone objective functions, existing approximation algorithms are inefficient relative to the fast algorithms developed for the case of monotone objectives. In this paper, we prove that SG (with slight modification) can achieve almost $1/4$-approximation guarantees in expectation in linear time even if objective functions are non-monotone. Our result provides a constant-factor approximation algorithm with the fewest oracle queries for non-monotone submodular maximization with a cardinality constraint. Experiments validate the performance of (modified) SG. 
	\end{abstract}

	\newcommand{\f}{f}
\newcommand{\fdel}[2]{\f_{{#2}}({#1})}
\newcommand{\As}{A}
\newcommand{\Aso}{\As^*}
\newcommand{\Rs}{R}
\newcommand{\s}{s}
\newcommand{\n}{n}
\renewcommand{\k}{k}
\renewcommand{\a}{a}

\newcommand{\clip}[1]{{[{#1}]_+}}

\newcommand{\h}[2]{h_{#1}({#2})}
\newcommand{\hek}{\h{\epsilon}{\k}}

\newcommand{\bs}{\backslash}
\newcommand{\sg}{SG}
\newcommand{\si}{\s}
\newcommand{\sj}{\s}

\newcommand{\Thetarm}{\mathrm{\Theta}}
\newcommand{\Thetarmp}[1]{\mathrm{\Theta}\left({#1}\right)}
\newcommand{\Omegarm}{\mathrm{\Omega}}
\newcommand{\Omegarmp}[1]{\mathrm{\Omega}\left({#1}\right)}
\newcommand{\Op}[1]{O\left({#1}\right)}

\newcommand{\e}{\mathrm{e}}

\section{INTRODUCTION}
We consider the following submodular function maximization problem with a cardinality constraint: 
\begin{align}\label{problem:main}
\maximize_{S\subseteq V}  \quad \f(S) 
\qquad
\subto\quad |S|\le\k,
\end{align}
where 
$V$ is a finite ground set of $n$ elements, 
$\f:2^V\to\R$ 
is a non-negative submodular function, 
and $\k$ ($\le n$) is a positive integer. 
As is conventionally done, 
we assume the value oracle model 
(i.e., $\f(\cdot)$ is a black-box function)
and discuss 
the complexity of algorithms in terms of the number of oracle queries, which we call the oracle complexity. 
Since the evaluation of $f$ is often expensive, to develop oracle-efficient algorithms has been an important research subject. 

For the case where $\f$ is monotone, 
the standard greedy algorithm 
can achieve a ($1-1/\e$)-approximation guarantee 
with $O(\k\n)$ queries \citep{nemhauser1978analysis}; 
this, however, is often too costly when applied to practical large-size instances.  
To deal with such large instances, 
various fast algorithms have been developed \citep{badanidiyuru2014fast,wei2014fast}. 
The stochastic greedy algorithm (\sg) \citep{mirzasoleiman2015lazier}  
is one such algorithm: 
In each iteration, 
instead of finding the element with the maximum marginal gain at the cost of up to $\n$ queries, 
we sample (roughly) ${\frac{\n}{\k}\ln\frac{1}{\epsilon}}$ elements uniformly at random, where $\epsilon\in(\e^{-\k},1)$, and choose the element with the largest marginal gain 
out of the sampled elements. 
\sg\ requires about $\n\ln\frac{1}{\epsilon}$ oracle queries in total, 
and it is known to achieve a ($1 -{1}/{\e} - \epsilon$)-approximation guarantee 
if $\f$ is monotone. 
Thanks to its simplicity, efficiency, strong guarantee, 
and 
high empirical performance, 
\sg\ has been used in various studies \citep{song2017deep,hasemi2018randomized}. 

\begin{table*}[htb]
	\caption{Comparison of Fast Algorithms for Non-monotone Submodular Maximization with Cardinality Constraint. 
	}
	\centering
	\begin{tabular}{ccclc}
		\toprule
		& Approximation ratio & Oracle complexity & &  Remark \\
		\midrule
		\multirow{2}{*}{Our result} 
		&
		\multirow{2}{*}{$\frac{1}{4}(1-\delta)^2$}
		&
		$\n\ln2 + \n\delta\frac{\k}{\k-1}$ 
		&(expectation) 
		& \multirow{2}{*}{Randomized}\\
		&
		&
		$\max\{\n, \k + \frac{2\k}{\delta}\}\times \ln2 + \k$
		&(worst case)
		& \\
		\multirow{2}{*}{\citet{buchbinder2017comparing}}
		&
		\multirow{2}{*}{${1}/{\e} - \epsilon$} 
		&
		\multirow{2}{*}{$\Op{\frac{\n}{\epsilon^2}\ln\frac{1}{\epsilon}}$}
		&
		&
		\multirow{2}{*}{Randomized}
		\\
		& &  & 
		\\
		\multirow{1}{*}{\citet{kuhnle2019interlaced}}			
		& 
		\multirow{1}{*}{${1}/{4}-\epsilon$}
		&
		\multirow{1}{*}{$\Op{\frac{\n}{\epsilon}\ln\frac{\n}{\epsilon}}$}
		&
		&
		\multirow{1}{*}{Deterministic}
		\\
		\bottomrule
		\label{table:}
	\end{tabular}
\end{table*}

Non-monotone submodular functions also appear in many practical scenarios:  
sensor placement \citep{krause2008near}, 
document summarization \citep{lin2010multi}, 
feature selection \citep{iyer2012algorithms}, 
and 
recommendation \citep{mirzasoleiman2016fast}.  
Unfortunately, the problem becomes much harder 
if $\f$ is non-monotone; 
for example, 
the approximation ratio of the greedy algorithm can become 
arbitrarily poor (at most $1/\k$-approximation) 
in general as in \citep[Appendix H.1]{pan2014parallel}. 
Although various constant-factor approximation algorithms 
for non-monotone objectives have been developed \citep{buchbinder2014submodular,buchbinder2017comparing,kuhnle2019interlaced}, 
they often require much more oracle queries than the aforementioned 
fast algorithms developed for monotone objectives, including \sg. 
Therefore,  
non-monotone submodular maximization with a cardinality constraint 
is currently awaiting oracle-efficient constant-factor approximation algorithms.

\subsection{Our Contribution}
We prove approximation guarantees of (modified) \sg\ for 
non-monotone objective functions, 
thus providing oracle-efficient approximation algorithms for non-monotone submodular maximization with a cardinality constraint. 
Below we detail our contributions: 
\begin{itemize}
	\item 
	Assuming $\n\ge3\k$, 
	we prove that \sg\ can achieve a 
	$\frac{1}{4} \left( 1 - 
	2\cdot\frac{\k-1}{\n-\k} \right)^2$-approximation 
	guarantee in expectation by setting $\epsilon$ at $\frac{1}{2}+ \frac{\k-1}{\n-\k}$. 
	Namely, 
	if $\n\gg\k$, 
	\sg\ can achieve an approximation ratio close to $1/4$ with about $\n\ln2$ queries. 
	
	\item We develop modified \sg\  
	such that the sample size in each iteration is also stochastic. 
	The resulting algorithm achieves a $\frac{1}{4}(1-\delta)^2$-approximation guarantee in expectation. 
	The expected and worst-case oracle complexities are bounded by 
	$\n\ln2 + \n\delta\frac{\k}{\k-1}=O(n)$ and 
	$\max\{\n, \k + \frac{2\k}{\delta}\}\times \ln2 + \k\le O(\n/\delta)$, 
	respectively. 
	Namely, modified SG is  a randomized linear-time constant-factor approximation algorithm. 
	As will be discussed in \Cref{subsec:related}, 
	this result provides a constant-factor approximation algorithm with the fewest oracle queries. 
	
	\item Experiments confirm the efficiency and high performance of (modified) \sg; 
	they run much faster and require far fewer queries 
	than existing algorithms while achieving comparable objective values. 
	The results demonstrate that 
	we can use (modified) \sg\ as practical and theoretically guaranteed algorithms even for non-monotone objectives. 
\end{itemize}

Note, however, that the approximation guarantees are required to hold only in expectation; 
the worst-case approximation ratio can be arbitrarily bad. 
This is why the constant-factor approximation guarantee with 
oracle queries 
possibly fewer than $n$, 
which may be counter-intuitive at first glance, is possible.

\subsection{Related Work}\label{subsec:related} 
\sg\ was proposed by~\citet{mirzasoleiman2015lazier} as an accelerated 
version of the well-known greedy algorithm \citep{nemhauser1978analysis} 
for monotone submodular maximization with a cardinality constraint. 
\citet{hassidim2017robust} studied a variant of \sg\ 
for monotone objectives and 
proved a guarantee that holds with a high probability. 
Guarantees of \sg\ for monotone set functions with approximate submodularity 
have also been widely studied \citep{khanna2017approximation,hasemi2018randomized,veciana2019stochasticgreedy}. 
\citet{harshaw2019submodular} studied \sg\ 
for maximizing set functions written as $f = g - c$, 
where $g$ is monotone weakly submodular and $c$ 
is non-negative modular; 
while $f$ can be non-monotone, 
they do not consider the whole class of non-monotone submodular functions 
and their approximation guarantee cannot be written with a 
multiplicative factor unlike our results. 

Constrained non-monotone submodular maximization has been extensively studied
\citep{lee2010maximizing,
	gupta2010constrained,
	feldman2011unified}. 
For the cardinality-constrained case, 
\citet{buchbinder2014submodular} proposed the random greedy algorithm, 
which behaves differently than SG. 
Specifically, 
it chooses an element uniformly at random from 
the top-$\k$ most beneficial elements in each iteration. 
While it achieves a $1/\e$-approximation guarantee, 
its oracle complexity is $\Op{\k\n}$, 
which is as costly as the standard greedy algorithm.  
They also achieved the best approximation ratio, 
${1}/{\e} + 0.004$, by combining the random greedy 
and continuous double greedy algorithms.  
\citet{buchbinder2018deterministic} derandomized the random greedy algorithm 
and achieved a ${1}/{\e}$-approximation guarantee with $O(\k^2\n)$ oracle queries; 
${1}/{\e}$ is the best ratio achieved by deterministic algorithms.  
As regards hardness results, 
\citet{vondrak2013symmetry} proved that 
to improve a ${1}/{2}$-approximation guarantee 
requires exponentially many queries when $\k={\n}/{2}$. 
For the case of $\k=o(\n)$, 
\citet{gharan2011submodular} 
proved 
a stronger hardness of $0.491$-approximation.

Regarding oracle-efficient algorithms, 
\citet{buchbinder2017comparing} proposed the random sampling algorithm (RS), 
which achieves a (${1}/{\e}-\epsilon$)-approximation 
with $\Op{\frac{n}{\epsilon^2}\ln\frac{1}{\epsilon}}$ 
oracle queries; 
to the best of our knowledge, 
this is the only existing linear-time constant-factor approximation algorithm.     
More precisely,
RS requires at least $\frac{8n}{\epsilon^2}\ln\frac{2}{\epsilon}$ queries; 
hence, to obtain a non-negative approximation ratio, 
we need at least $8\e^2n\ln(2\e)\ge100n$ queries.  
On the other hand, 
the expected and worst-case oracle complexities of 
the modified \sg\ are at most 
$\n\ln2 + \n\delta\frac{\k}{\k-1}$ and 
$\max\{\n, \k + \frac{2\k}{\delta}\}\times \ln2 + \k$, respectively.  
Therefore, 
taking the constant factors into account,  
\sg\ is far faster than RS. 
In \Cref{sec:experiments}, 
we experimentally confirm that this gap is crucial in practice. 	 
\citet{buchbinder2017comparing} also developed another algorithm that achieves a (${1}/{\e}-\epsilon$)-approximation guarantee 
with $\Op{\k\sqrt{\frac{\n}{\epsilon}\ln\frac{\k}{\epsilon}} + \frac{\n}{\epsilon}\ln\frac{\k}{\epsilon}}$ 
oracle queries in expectation. 
Since $\k=\Thetarm(\n)$ in general, it is more costly than \sg.
Recently, \citet{kuhnle2019interlaced} proposed a deterministic  (${1}/{4}-\epsilon$)-approximation algorithm with 
$O(\frac{\n}{\epsilon}\ln\frac{\n}{\epsilon})$ queries, 
which is the best oracle complexity among those of deterministic algorithms.  
Note that 
it is also slower than \sg\ due to the presence of the $\ln{\n}$ factor.   
\Cref{table:} 
compares the above results and ours. 	 

We remark that our work is different from 
\citep{qian2018approximation,ji2020stochastic}, 
which are seemingly overlapping with ours.  
Their algorithms for non-monotone objectives are 
not \sg-style ones but variants of the aforementioned random greedy algorithm. 
Hence, 
unlike \sg\ and the above efficient algorithms, 
their algorithms generally require $O(kn)$ queries.  
Approximation algorithms for non-monotone submodular maximization with more general constraints 
have also been studied \citep{mirzasoleiman2016fast,feldman2017greed}. 
If those algorithms are applied to the cardinality-constrained case, 
we need $\Omegarm(\k\n)$ queries in general.

Recently, 
parallel non-monotone submodular maximization algorithms have been widely studied 
\citep{balkanski2018nonmonotone,ene2019submodular,fahrbach2019non-monotone}. 
Unlike us, 
they are interested in a different complexity framework 
called the adaptive complexity, 
which is defined with the number of 
sequential rounds required 
when polynomially many oracle queries can be executed in parallel.
As summarized in \citep{fahrbach2019non-monotone}, 
such parallel algorithms require more than 
$\Omegarm(n)$ oracle queries; 
among them, 
a $(0.039-\epsilon)$-approximation algorithm of \citep{fahrbach2019non-monotone} requires the fewest 
queries, $O(\frac{n}{\epsilon^2}\ln k)$, in expectation. 
Unlike those algorithms,  
\sg\ requires only $O(n)$ queries in expectation.

\subsection{Notation and Definitions}
Given a set function $\f:2^V\to\R$, 
we define $\fdel{Y}{X}\coloneqq
\f(X\cup Y) - \f(X)$ for any $X,Y\subseteq V$. 
We sometimes abuse the notation and regard $v\in V$ as a subset 
(e.g., 
we use $\fdel{v}{X}$ instead of $\fdel{\{v\}}{X}$). 
We say $\f$ is 
non-negative if $\f(X)\ge0$ for any $X\subseteq V$, 
monotone if $\fdel{v}{X}\ge0$ for any $X\subseteq V$ and $v\notin X$, 
normalized if $\f(\emptyset)=0$, 
and 
submodular if 
$\f(X) + \f(Y) \ge \f(X\cup Y) + \f(X\cap Y)$ for any $X,Y\subseteq V$, 
which is also equivalently characterized by the following 
diminishing return property: 
$\fdel{v}{X}\ge\fdel{v}{Y}$ for any $X\subseteq Y$ 
and $v\notin Y$. 
In this paper, 
all set functions are assumed to be non-negative and submodular 
(not necessarily monotone and normalized) 
unless otherwise specified.  
In what follows, 
we use $\Aso$ to denote an optimal solution to 
problem \eqref{problem:main}.

\subsection{Organization}
\Cref{sec:sg} reviews the details of \sg\ and the proof for the case of monotone objectives. 
In \Cref{sec:nonmonotone} we prove the approximation guarantees of 
(modified) \sg\ for the case of non-monotone objectives. 
\Cref{sec:experiments} presents the experimental results. 
\Cref{sec:conclusion} concludes this paper. 
All missing proofs are presented in the appendix.

\section{STOCHASTIC GREEDY AND PROOF FOR MONOTONE CASE}\label{sec:sg}
We here review the details of \sg\ and the proof for the case of monotone objectives~\citep{mirzasoleiman2015lazier}, 
which will help us to understand the main discussion presented in \Cref{sec:nonmonotone}. 
\begin{algorithm}[tb]
	\caption{Stochastic Greedy (\sg)}  \label{algorithm}
	\begin{algorithmic}[1]
		\State $\As_0\gets\emptyset$
		\For{$i=1,\dots,k$}
		\State Get $\Rs$ by sampling $\ceil{\si}$ elements from $V\bs\As_{i-1}$
		\State $\a_i\gets\argmax_{\a\in\Rs}\fdel{\a}{\As_{i-1}}$
		\If{$\fdel{\a_i}{\As_{i-1}}>0$} $\As_i\gets\As_{i-1}\cup\{\a_i\}$ \label{step:if}
		\Else{} $\As_i\gets\As_{i-1}$ \label{step:asis}
		\EndIf			
		\EndFor
		\State \Return $\As_\k$
	\end{algorithmic}
\end{algorithm} 

Let $\si\coloneqq\frac{\n}{\k}\ln\frac{1}{\epsilon}$. 
In each iteration of \sg\ (\Cref{algorithm}), 
we choose the best element from $\ceil{\s}$ elements sampled uniformly at random from $V\bs\As_{i-1}$ without replacement. 
We remark that 
\Cref{algorithm} is slightly different from the original \sg\ \citep{mirzasoleiman2015lazier}: 
since the marginal gain can be negative due to the lack of monotonicity, 
we let \Cref{algorithm} to reject elements with non-positive marginal gains 
as in Steps~\ref{step:if} and \ref{step:asis} 
(elements with zero gains are rejected to simplify the discussion in \Cref{subsec:sged}). 
As is usual with the proofs of greedy-style algorithms, 
we first consider lower bounding the marginal gain of each iteration as follows: 
\begin{restatable}[cf. \citep{mirzasoleiman2015lazier}]{lem}{marginal}\label{lem:marginal}
	If $f$ is non-negative and submodular, 
	for $i=1,\dots,\k$, 
	we have
	\begin{align+}
		\E[\f(\As_{i}) - \f(\As_{i-1})]
		\ge{} 
		\frac{1-\epsilon}{\k}
		\E[\fdel{\Aso}{\As_{i-1}}].\nonumber
	\end{align+}
\end{restatable}
While \citet{mirzasoleiman2015lazier} proved this lemma 
implicitly relying on the monotonicity of $\f$, 
we can prove it even if $f$ is non-monotone 
due to the non-negativity of $\E[\f(\As_{i}) - \f(\As_{i-1})]$, 
which is obtained from Steps~\ref{step:if} and \ref{step:asis} 
(see, \Cref{a_sec:lemma1} for the proof). 
This non-monotone version of the lemma will play an important role 
in the proof of our main result presented in \Cref{sec:nonmonotone}.

We now see how to prove the ($1-1/\e-\epsilon$)-approximation guarantee 
of \sg\ for the case of monotone objectives. 
Assume that $\f$ is monotone and normalized. 
We have $\E[\f(\Aso\cup\As_{i-1})]\ge\f(\Aso)$ due to the monotonicity, 
and thus \Cref{lem:marginal} implies   
\begin{align}
\E[\f(\As_{i}) - \f(\As_{i-1})]
\ge{} 
\frac{1-\epsilon}{\k}
(\f(\Aso) - \E[\f(\As_{i-1})]).  
\end{align}
By using this inequality for $i=1,\dots,\k$ 
and $\f(\emptyset)=0$, 
we obtain the desired result as follows: 
\begin{align}
\E[\f(\As_{\k})] 
\ge 
\f(\Aso)- \left( 1 - \frac{1-\epsilon}{\k} \right)^\k(\f(\Aso) - \f(\emptyset)) 
\ge 
\left(1-\frac{1}{\e^{1-\epsilon}}\right)\f(\Aso)
\ge 
\left(1-\frac{1}{\e}- \epsilon \right)\f(\Aso).
\end{align}
In the above proof, 
the inequality, $\E[\f(\Aso\cup\As_{i-1})] \ge \f(\Aso)$, 
obtained with the monotonicity, played a key role. 
As will be shown in \Cref{subsec:sg}, 
we can derive a variant of the inequality 
for non-monotone $\f$ by using the randomness of \sg, 
which enables us to prove approximation guarantees 
without the monotonicity. 

\renewcommand{\c}{c}

\section{PROOF FOR NON-MONOTONE CASE}\label{sec:nonmonotone}

We present approximation guarantees of (modified) \sg\ for non-monotone objectives. 
In \Cref{subsec:sg}, 
we prove the $\frac{1}{4} \left( 1 - 
2\cdot\frac{\k-1}{\n-\k} \right)^2$-approximation guarantee of \sg, 
and in \Cref{subsec:sged} we prove the $\frac{1}{4}(1-\delta)^2$-approximation guarantee  of modified \sg.  

\subsection{$\frac{1}{4} \left( 1 - 2\cdot\frac{\k-1}{\n-\k} \right)^2$-approximation of \sg}\label{subsec:sg}

We here make the following assumption: 
\begin{assump}\label{assump:hek}
	We assume that $\k\ge2$ and
	$\n \ge 3\k$ hold 
	and that 
	$\epsilon$ is set so as to satisfy $1/\e\le\epsilon<1$.  
\end{assump}

The first assumption, $\k\ge2$, is natural 
since, 
if $\k=1$, 
an $\alpha$-approximation guarantee 
($\forall\alpha\in[0,1]$) can be achieved in expectation by examining 
$\ceil{\alpha\n}$ elements, 
which means any approximation ratio can be achieved 
in linear time.    
Hence we assume $\k\ge2$ in what follows. 
The second assumption, $\n\ge3\k$, 
will be removed in \Cref{subsec:sged}. 
The third assumption, $1/\e\le\epsilon<1$, 
can be easily satisfied since $\epsilon$ is a controllable 
input.


We derive
a variant of $\E[\f(\Aso\cup\As_{i-1})]\ge\f(\Aso)$ for non-monotone $\f$. 
To this end, we use the following lemma: 
\begin{lem}[\citet{buchbinder2014submodular}, Lemma 2.2]\label{lem:buchbinder}
	Let $g:2^V\to\R$ be submodular. 
	Denote by $\As(p)$ 
	a random subset of $\As\subseteq V$ 
	where each element appears with 
	a probability of at most $p$ (not necessarily independently). 
	Then, $\E[g(\As(p))] \ge (1-p) g(\emptyset)$.	
\end{lem}
Namely, 
if $\As_{i-1}$ includes each $a\in V$ with a probability of at most $p$, 
then $\E[\f(\Aso\cup\As_{i-1})] \ge (1-p)\f(\Aso)$ holds. 
Below we upper bound $p$ by leveraging the randomness of \sg\ 
and prove the following lemma: 

\begin{lem}\label{lem:opt}
	Assume that $1/\e\le\epsilon<1$ holds. 
	Then,   
	for $i=0,\dots,\k$, we have 
	\begin{align}\label{eq:variant}
	\E[\f(\Aso\cup\As_{i})] 
	\ge
	\left( 1-\frac{1}{\k}\ln\frac{1}{\epsilon} - \frac{2}{\n-\k}\right)^{i} 
	\f(\Aso). 
	\end{align}
\end{lem}

\begin{proof}[Proof of \Cref{lem:opt}]
	If $i=0$, the lemma holds since $\As_0=\emptyset$. 
	Below we assume $i\ge1$. 
	In the $i$-th iteration, 
	conditioned on $\As_{i-1}$, 
	each $a\in V\bs\As_{i-1}$ stays outside of $\As_{i}$ 
	with a probability of at least $1-\frac{\ceil{\si}}{|V\bs\As_{i-1}|}$. 
	Hence, 
	after $i$ iterations ($i=1,\dots,\k$), 
	each $\a\in V$ stays outside of $\As_{i}$ with a probability of at least 
	\begin{align}
	\prod_{j=1}^i
	\left( 1 - \frac{\ceil{\sj}}{|V\bs\As_{j-1}|} \right)
	\ge{}
	\prod_{j=1}^i
	\left( 1 - \frac{\sj+1}{\n-\k} \right)
	={}
	\left( 1-\frac{1}{\k}\ln\frac{1}{\epsilon} - \frac{1+\ln\frac{1}{\epsilon}}{\n-\k}  \right)^{i}.
	\end{align}
	Therefore, from $\epsilon\ge1/\e$, we obtain  
	\[
	\Pr[\a\in\As_{i}] \le 
	1 - \left( 1-\frac{1}{\k}\ln\frac{1}{\epsilon} - \frac{2}{\n-\k} \right)^{i}.
	\]
	We define $g(\As)\coloneqq\f(\As\cup\Aso)$, 
	which we can easily confirm to be submodular. 
	From \Cref{lem:buchbinder}, 
	we obtain 
	\begin{align}
	\E[\f(\Aso\cup\As_{i})] 
	={}
	\E[g(\As_{i})]
	\ge{}
	\left( 1-\frac{1}{\k}\ln\frac{1}{\epsilon} - \frac{2}{\n-\k} \right)^{i}
	\E[g(\emptyset)] 
	={}
	\left( 1-\frac{1}{\k}\ln\frac{1}{\epsilon} - \frac{2}{\n-\k} \right)^{i}
	\f(\Aso). 
	\end{align}
	Hence the lemma holds. 
\end{proof} 

We then consider lower bounding the RHS of 
\eqref{eq:variant} 
for $i=\k-1$. 
Intuitively, if $\n\gg\k$ and the $\frac{2}{\n-\k}$ term is ignorably small, 
the RHS can be lower bounded by $\epsilon\f(\Aso)$
since 
$\left( 1- \frac{1}{\k}\ln\frac{1}{\epsilon} \right)^{\k-1}\approx \e^{-\ln\frac{1}{\epsilon}}=\epsilon$.  
By evaluating the RHS more carefully using \Cref{assump:hek}, 
we can obtain the following lemma (proof is provided in \Cref{a_sec:lemma4}): 	

\begin{restatable}{lem}{eps}\label{lem:eps}
	If \Cref{assump:hek} holds, we have
	\begin{align+}
		\left( 1-\frac{1}{\k}\ln\frac{1}{\epsilon} - \frac{2}{\n-\k}\right)^{\k-1}
		\ge
		\epsilon - 2\cdot\frac{\k-1}{\n-\k}. \nonumber
	\end{align+}
\end{restatable}

We are now ready to prove the approximation guarantee of \sg\ 
for the case of non-monotone objectives. 

\begin{thm}\label{thm:nonmonotone}
	Let $\As$ be the output of \Cref{algorithm}. 
	If \Cref{assump:hek} holds, we have  
	\begin{align}
	\E[\f(\As)]
	\ge
	\left( \epsilon - 2\cdot\frac{\k-1}{\n-\k} \right)
	(1-\epsilon)
	\f(\Aso).
	\end{align}
	By setting $\epsilon=\frac{1}{2}+ \frac{\k-1}{\n-\k}$, 
	we obtain 
	\[
	\E[\f(\As)]
	\ge
	\frac{1}{4}
	\left( 1 - 2\cdot\frac{\k-1}{\n-\k} \right)^2
	\f(\Aso).
	\]
\end{thm}
Note that $1/\e\le\frac{1}{2}+\frac{\k-1}{\n-\k}<1$ holds 
since $\n\ge3\k$. 
The following proof is partly inspired by the technique used in \citep{buchbinder2014submodular}.

\begin{proof}[Proof of \Cref{thm:nonmonotone}]
	We prove 
	\begin{align}\label{eq:target}
	\begin{aligned}
	\frac{\E[\f(\As_{i})]}{\f(\Aso)}
	\ge{}
	\frac{i}{\k}\left( 1-\frac{1}{\k}\ln\frac{1}{\epsilon} 
	-\frac{2}{\n-\k}  \right)^{i-1}(1-\epsilon)
	\end{aligned}
	\end{align} 
	for $i=0,\dots,\k$
	by induction. 
	If $i=0$, the RHS of \eqref{eq:target} becomes $0$, 
	and so the inequality holds due to the non-negativity of $\f$. 
	Assume that \eqref{eq:target} holds for every $i^\prime=0,\dots,i-1$. 
	Then we have 
	\begin{align}
	\E[\f(\As_{i})]
	={}&
	\E[\f(\As_{i-1})] + \E[\f(\As_{i}) - \f(\As_{i-1})]
	\\
	\ge{}&
	\E[\f(\As_{i-1})] 
	+
	\frac{1-\epsilon}{\k}
	\E\left[ \f(\Aso\cup\As_{i-1}) - \f(\As_{i-1}) \right]	
	\\
	\tag{\Cref{lem:marginal}}
	\\
	={}&
	\left( 1 - \frac{1-\epsilon}{\k} \right)
	\E[\f(\As_{i-1})] 
	+
	\frac{1-\epsilon}{\k}
	\left( 1-\frac{1}{\k}\ln\frac{1}{\epsilon} - \frac{2}{\n-\k} \right)^{i-1}
	\f(\Aso)
	\\
	\tag{\Cref{lem:opt}}
	\\
	\ge{}&
	\left( 1-\frac{1}{\k}\ln\frac{1}{\epsilon} - \frac{2}{\n-\k} \right)
	\E[\f(\As_{i-1})] 
	 +
	\frac{1-\epsilon}{\k}
	\left( 1-\frac{1}{\k}\ln\frac{1}{\epsilon} - \frac{2}{\n-\k} \right)^{i-1}
	\f(\Aso)
	\\
	\tag{$1 - \epsilon \le \ln\frac{1}{\epsilon} \le \ln\frac{1}{\epsilon} + \frac{2\k}{\n-\k}$}
	\\
	\ge{}&
	\left( 1-\frac{1}{\k}\ln\frac{1}{\epsilon} - \frac{2}{\n-\k} \right)
	\times 
	\frac{i-1}{k}
	\left( 1-\frac{1}{\k}\ln\frac{1}{\epsilon} - \frac{2}{\n-\k} \right)^{i-2}
	(1-\epsilon)
	\f(\Aso)
	\\
	& +
	\frac{1-\epsilon}{\k}
	\left( 1-\frac{1}{\k}\ln\frac{1}{\epsilon} - \frac{2}{\n-\k} \right)^{i-1}
	\f(\Aso)
	\\
	\tag{Assumption of induction}
	\\
	={}&
	\frac{i}{\k}
	\left( 1-\frac{1}{\k}\ln\frac{1}{\epsilon} -\frac{2}{\n-\k}  \right)^{i-1}
	(1-\epsilon)
	\f(\Aso).
	\end{align}
	Hence \eqref{eq:target} holds for $i=0,\dots,\k$;  
	for $i=\k$, we have 
	\begin{align}
	\E[\f(\As_{\k})]
	\ge
	\left( 1-\frac{1}{\k}\ln\frac{1}{\epsilon} -\frac{2}{\n-\k}  \right)^{\k-1}
	(1-\epsilon)
	\f(\Aso).
	\end{align}
	Finally, by using \Cref{lem:eps}, 
	we can lower bound the approximation ratio by $\left( \epsilon - 2\cdot\frac{\k-1}{\n-\k} \right)(1-\epsilon)$. 
\end{proof}

Note that, while \Cref{lem:marginal} motivates us to let $\epsilon$ to be small, 
the opposite is true regarding \Cref{lem:opt}. 
Thus we set $\epsilon\approx1/2$ to balance the effects of the two inequalities.

\newcommand{\D}{D}
\newcommand{\Abaro}{{\overline\Aso}}

\renewcommand{\N}{N}
\newcommand{\mi}{{m_i}}
\newcommand{\Mi}{{M_i}}

\renewcommand{\r}{r}
\newcommand{\sged}[2]{\sg-({#1},{#2})}

\newcommand{\hgd}{\text{H}}

\subsection{$\frac{1}{4}(1-\delta)^2$-approximation of Modified \sg}\label{subsec:sged}
As shown in \Cref{subsec:sg}, 
the approximation ratio of \sg\ becomes close to ${1}/{4}$ 
if $\n\gg\k$. 
In this section, we first consider improving the ratio by adding sufficiently many 
dummy elements to $V$. We then present modified \sg\ 
that can 
achieve a $\frac{1}{4}(1-\delta)^2$-approximation guarantee 
without using dummy elements explicitly. 

Let $\D$ be a set of dummy elements and $\Vbar=V\cup\D$; 
i.e., 
we have $\fdel{\a}{\As}=0$ 
for any $\As\subseteq \Vbar$ and $\a\in\D$.  
We add sufficiently many dummy elements to $V$ 
so that $\N\coloneqq|\Vbar|$ 
becomes equal to $\max\{\n, \k + \ceil{(2\k-1)/\delta}\}$, 
where $\delta\in(0,1)$ is an input parameter; 
smaller $\delta$ means that we add more dummy elements.  
Note that we have 
$\N\ge\k+2(\k-1)/\delta$ and $\N\ge\n$. 
We can also easily prove $\N\ge3\k$ by induction; 
this enables us to remove the second assumption of \Cref{assump:hek}.  

We now consider performing \sg\ on $\Vbar$.  
Let $\Abar$ be the output of \sg\ and $\Abaro=\argmax_{S\subseteq\Vbar:|S|\le\k}\f(S)$.  
Thanks to \Cref{thm:nonmonotone}, we have
\begin{align}
\E[\f(\Abar)]
\ge
\left( \epsilon - 2\cdot\frac{\k-1}{\N-\k} \right)
(1-\epsilon)
\f(\Abaro).
\end{align}
If we set $\epsilon=\frac{1}{2} + \frac{\k-1}{\N - \k}$, we obtain 
\begin{align}
\E[\f(\Abar)]
\ge 
\frac{1}{4}\left( 1 -  2\cdot\frac{\k-1}{\N - \k}\right)^2 \f(\Abaro)
\ge \frac{1}{4}(1-\delta)^2 \f(\Abaro)
\ge
\frac{1}{4}(1-\delta)^2 \f(\Aso), 
\end{align}
where 
the second inequality comes from $\N\ge\k+2(\k-1)/\delta$ and 
the last inequality comes from $V\subseteq\Vbar$. 
Furthermore, since 
no elements with zero marginal gains are added to the current solution in each iteration, 
$\Abar$ includes no elements in $\D$ 
(i.e., $\Abar\subseteq V$).
Therefore, 
$\Abar$ is a feasible $\frac{1}{4}(1-\delta)^2$-approximate solution. 

We then discuss the oracle complexity of performing \sg\ on $\Vbar$. 
In each iteration, we sample 
$\ceil{\sbar}$ elements to get $\Rs$, where 
$\sbar\coloneqq{\frac{\N}{\k}\ln\frac{1}{\epsilon}}$,  
and then we compute $\a_i=\argmax_{\a\in\Rs} \fdel{\a}{\As_{i-1}}$.  
Note that, 
if $\a\in\D$, 
we do not need to compute $\fdel{\a}{\As_{i-1}}$ since the value is always equal to $0$; 
i.e., any $\a\in\D$ is taken out of consideration.  
Therefore, 
only the number of elements belonging to $\Rs\cap V$ matters to the oracle complexity, 
which conforms to the hypergeometric distribution 
with 
a population of size ${|\Vbar\bs\As_{i-1}|}$, 
$\ceil{\sbar}$ draws, 
and 
${|V\bs\As_{i-1}|}$ targets. 
We denote the distribution by 
$\hgd(\ceil{\sbar}, {|V\bs\As_{i-1}|}, {|\Vbar\bs\As_{i-1}|})$.  
Note that its mean is bounded as follows: 
\begin{align}
\ceil{\sbar}\frac{|V\bs\As_{i-1}|}{|\Vbar\bs\As_{i-1}|}
\le
\left( \frac{\N}{\k}\ln\frac{1}{\epsilon} + 1\right)
\frac{\n}{\N-\k}
\le
\frac{\n}{\k}\ln\frac{1}{\epsilon} + \frac{\n}{\k-1}\delta, 
\end{align}
where the last inequality comes from $\N\ge\k+2(\k-1)/\delta$ and $\epsilon\ge1/\e$. 
Thus, the total oracle complexity is at most
$\n\ln\frac{1}{\epsilon} + \n\delta\frac{\k}{\k-1}$ in expectation, 
which can be 
arbitrarily close to $\n\ln\frac{1}{\epsilon}$ 
if $\delta$ is sufficiently small. 
Therefore, 
by setting $\epsilon=\frac{1}{2} + \frac{\k-1}{\N - \k}\ge \frac{1}{2}$, 
we can achieve a $\frac{1}{4}(1-\delta)^2$-approximation 
guarantee 
with at most $\n\ln2 + \n\delta\frac{\k}{\k-1}$ oracle queries in expectation.  Furthermore, 
the worst-case oracle complexity is also bounded by  
\[
\k\ceil{\sbar} 
=
\N\ln\frac{1}{\epsilon} + \k
\le
\max\left\{\n, \k+{\frac{2\k}{\delta}}\right\} \times \ln\frac{1}{\epsilon} + \k.
\]

\begin{algorithm}[tb]
	\caption{Modified \sg}  \label{algorithm:delta}
	\begin{algorithmic}[1]
		\State $\N\gets \max\{\n, \k+\ceil{{(2\k-1)}/{\delta}}\}$ and $\sbar\gets\frac{\N}{\k}\ln\frac{1}{\epsilon}$ 
		\State $\As_0\gets\emptyset$ 
		\For{$i=1,\dots,k$}
		\State Draw $\r\sim\hgd\left(\ceil{\sbar}, {|V\bs\As_{i-1}|}, {\N - |\As_{i-1}|} \right)$ 
		\State Get $\Rs$ by sampling $\r$ elements from $V\bs\As_{i-1}$
		\State $\a_i\gets\argmax_{\a\in\Rs}\fdel{\a}{\As_{i-1}}$
		\If{$\fdel{\a_i}{\As_{i-1}}>0$} $\As_i\gets\As_{i-1}\cup\{\a_i\}$
		\Else{} $\As_i\gets\As_{i-1}$
		\EndIf
		\EndFor
		\State \Return $\As_\k$
	\end{algorithmic}
\end{algorithm}

Finally, we see that no dummy elements are needed explicitly. 
As mentioned above, 
only the elements in $\Rs\cap V$ 
affects the behavior of \sg\ performed on $\Vbar$, 
and 
so an algorithm with the same behavior 
can be obtained by sampling $\Rs\subseteq V$ 
as follows: 
Draw 
$\r\in[0,\ceil{\sbar}]$ from the hypergeometric distribution 
$\hgd\left(\ceil{\sbar}, {|V\bs\As_{i-1}|}, {\N - |\As_{i-1}|} \right)$
and get $\Rs$ by sampling $\r$ elements uniformly at random from $V\bs\As_{i-1}$ without replacement. 
\Cref{algorithm:delta}, called modified \sg, presents the details.  
To conclude, 
we obtain the following result: 

\begin{thm}\label{thm:sged} 
	Fix $\epsilon$ and $\delta$ so that 
	$\epsilon\in[1/\e,1)$ and $\delta\in(0,\epsilon)$ hold, respectively.   
	Then, \Cref{algorithm:delta} outputs solution $A$
	that satisfies
	\begin{align}
	\E[\f(\As)]
	\ge
	\left( \epsilon - \delta \right)
	(1-\epsilon)
	\f(\Aso).
	\end{align}
	The expected and worst-case oracle complexities are at most
	$\n\ln\frac{1}{\epsilon} + \n\delta\frac{\k}{\k-1}$ 
	and 
	$\max\left\{\n, \k+{\frac{2\k}{\delta}}\right\} \times \ln\frac{1}{\epsilon} + \k$, 
	respectively. 
	If we set $\epsilon=\frac{1}{2} + \frac{\k-1}{\N - \k}$, 
	we have 
	\[
	\E[\f(\As)]
	\ge
	\frac{1}{4}
	\left( 1 - \delta \right)^2
	\f(\Aso).
	\]
	The expected and worst-case oracle complexities are bounded by $\n\ln2 + \n\delta\frac{\k}{\k-1}$ 
	and 
	$\max\left\{\n, \k+{\frac{2\k}{\delta}}\right\} \times \ln2 + \k$, 
	respectively.   
\end{thm}

Note that \Cref{algorithm:delta} can also achieve 
a ($1-\frac{1}{\e}-\epsilon$)-approximation guarantee 
if $\f$ is monotone since it is equivalent to \Cref{algorithm} performed on 
$\Vbar$ and the obtained solution is feasible as explained above.

\newcommand{\Ab}{{\text{$\mathbf{A}$}}}
\newcommand{\Abt}{{\text{$\mathbf{A}^\top$}}}
\newcommand{\Xb}{{\text{$\mathbf{X}$}}}
\newcommand{\Ib}{{\text{$\mathbf{I}$}}}
\newcommand{\msg}{MSG}

\section{EXPERIMENTS}\label{sec:experiments}	

\begin{figure}[tb]
	\centering
	\begin{minipage}[t]{.3\textwidth}
		\includegraphics[width=1.0\textwidth]{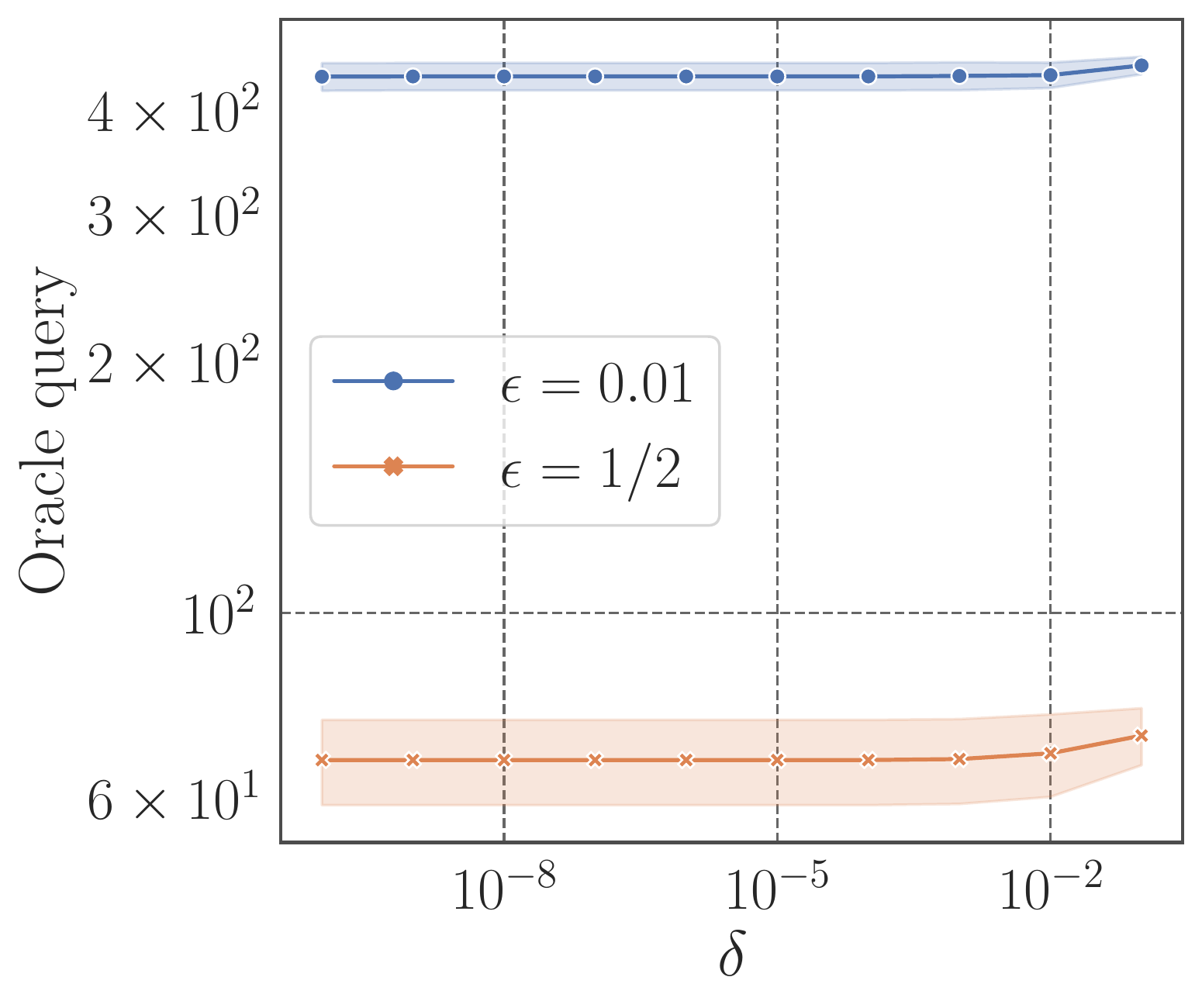}
		\subcaption{Oracle Query (semi-log)}
		\label{fig:delta_o}
	\end{minipage}
	\begin{minipage}[t]{.3\textwidth}
		\includegraphics[width=1.0\textwidth]{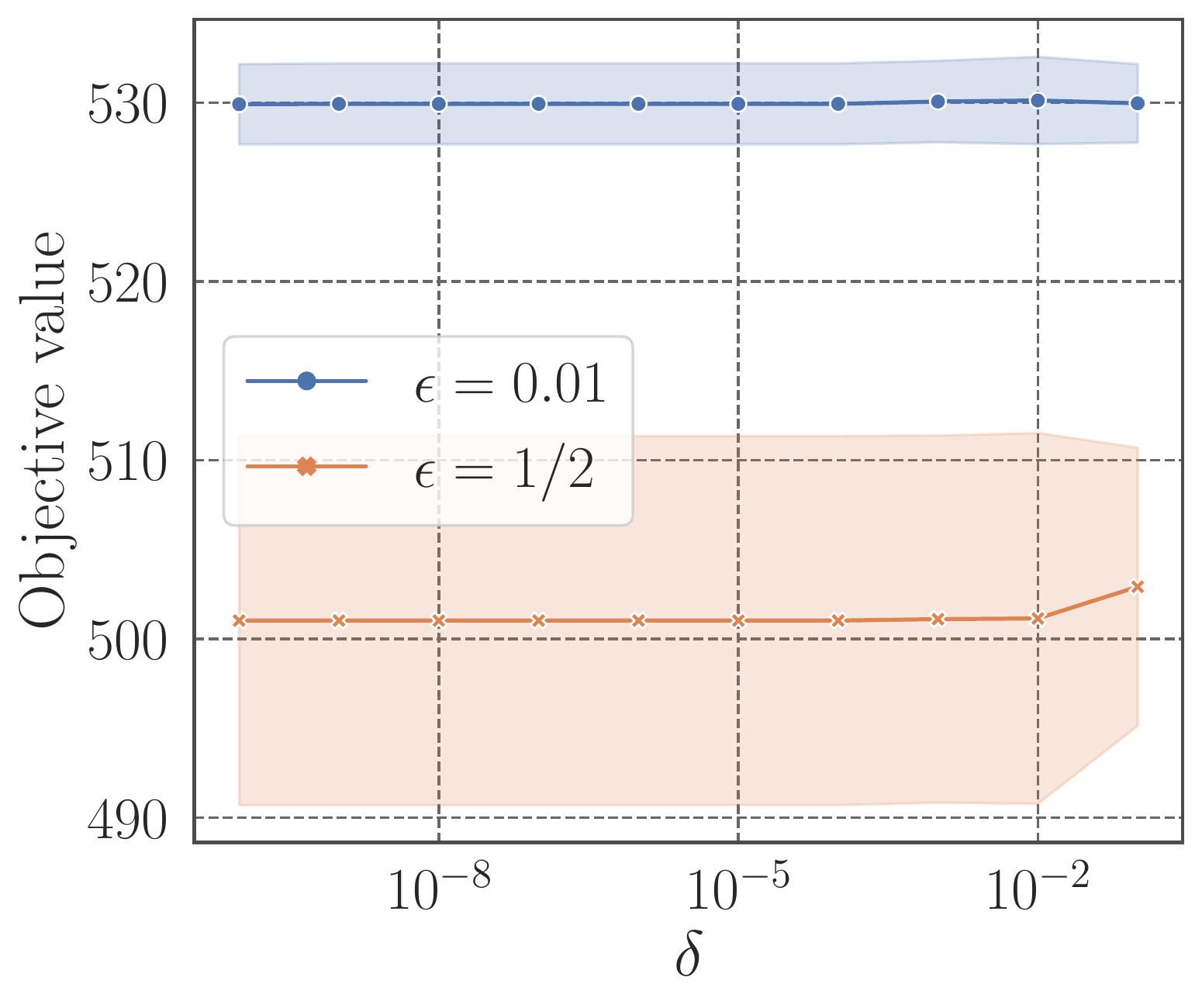}
		\subcaption{Objective Value}
		\label{fig:delta_v}
	\end{minipage}
	\caption{\msg\ Performance with Various $\delta$ Values.}
\end{figure}

\begin{figure*}[tb]
	\centering
	\begin{minipage}[t]{0.32\textwidth}
		\includegraphics[width=.88\textwidth]{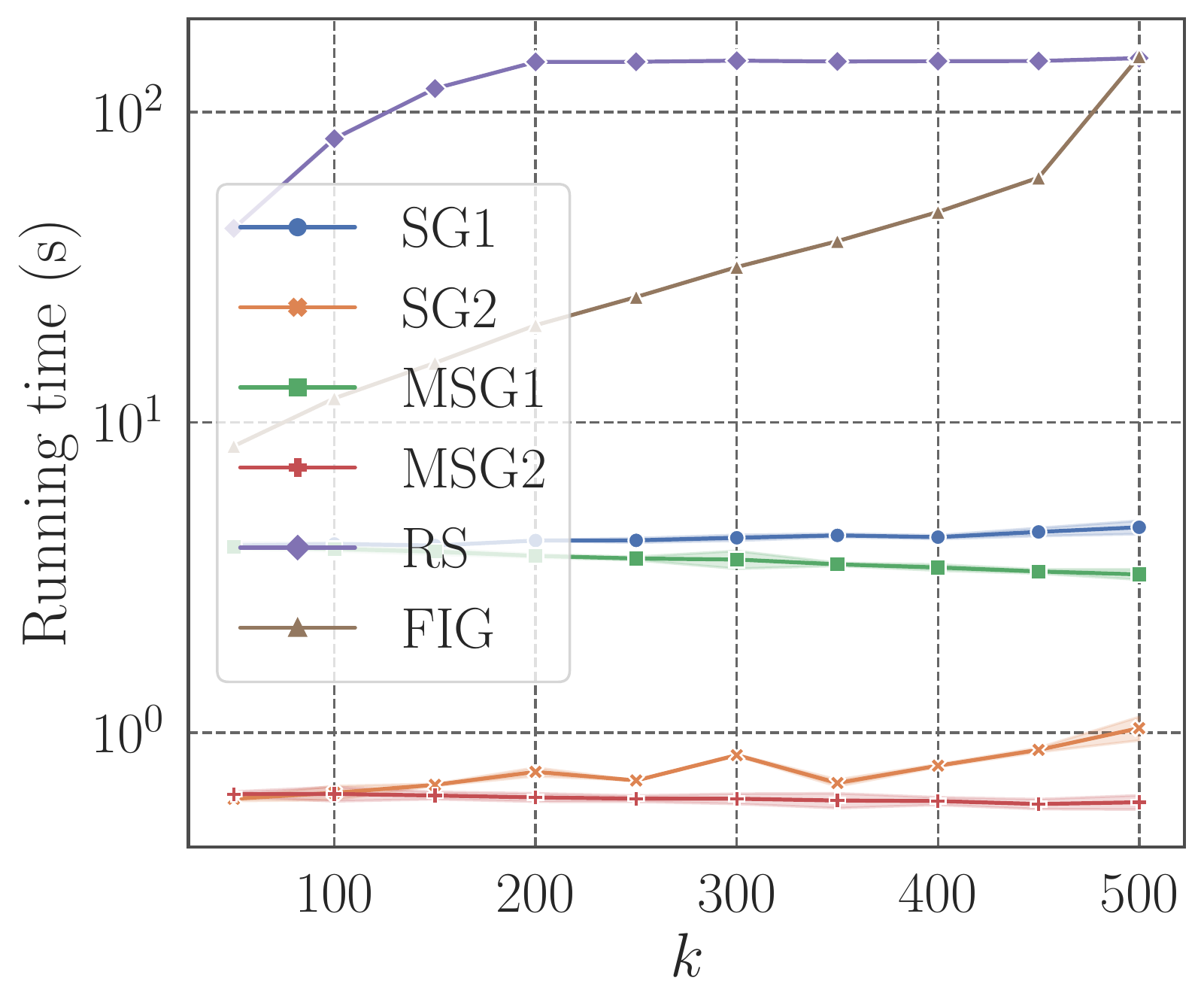}
		\subcaption{ER, Running Time (semi-log)}
		\label{fig:er_t}
	\end{minipage}
	\begin{minipage}[t]{0.32\textwidth}
		\includegraphics[width=.88\textwidth]{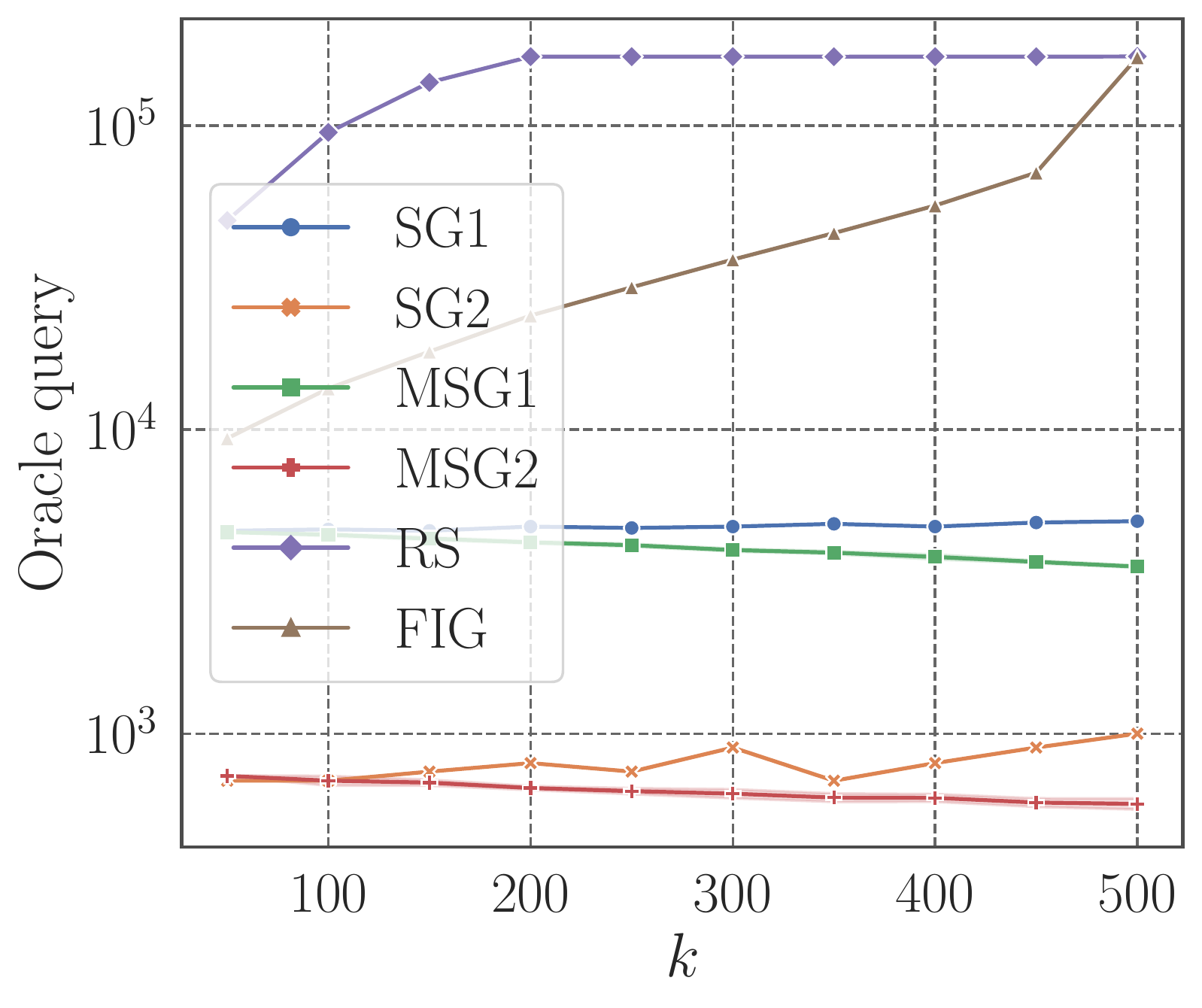}
		\subcaption{ER, Oracle Query (semi-log)}
		\label{fig:er_o}
	\end{minipage}
	\begin{minipage}[t]{0.32\textwidth}
		\includegraphics[width=.88\textwidth]{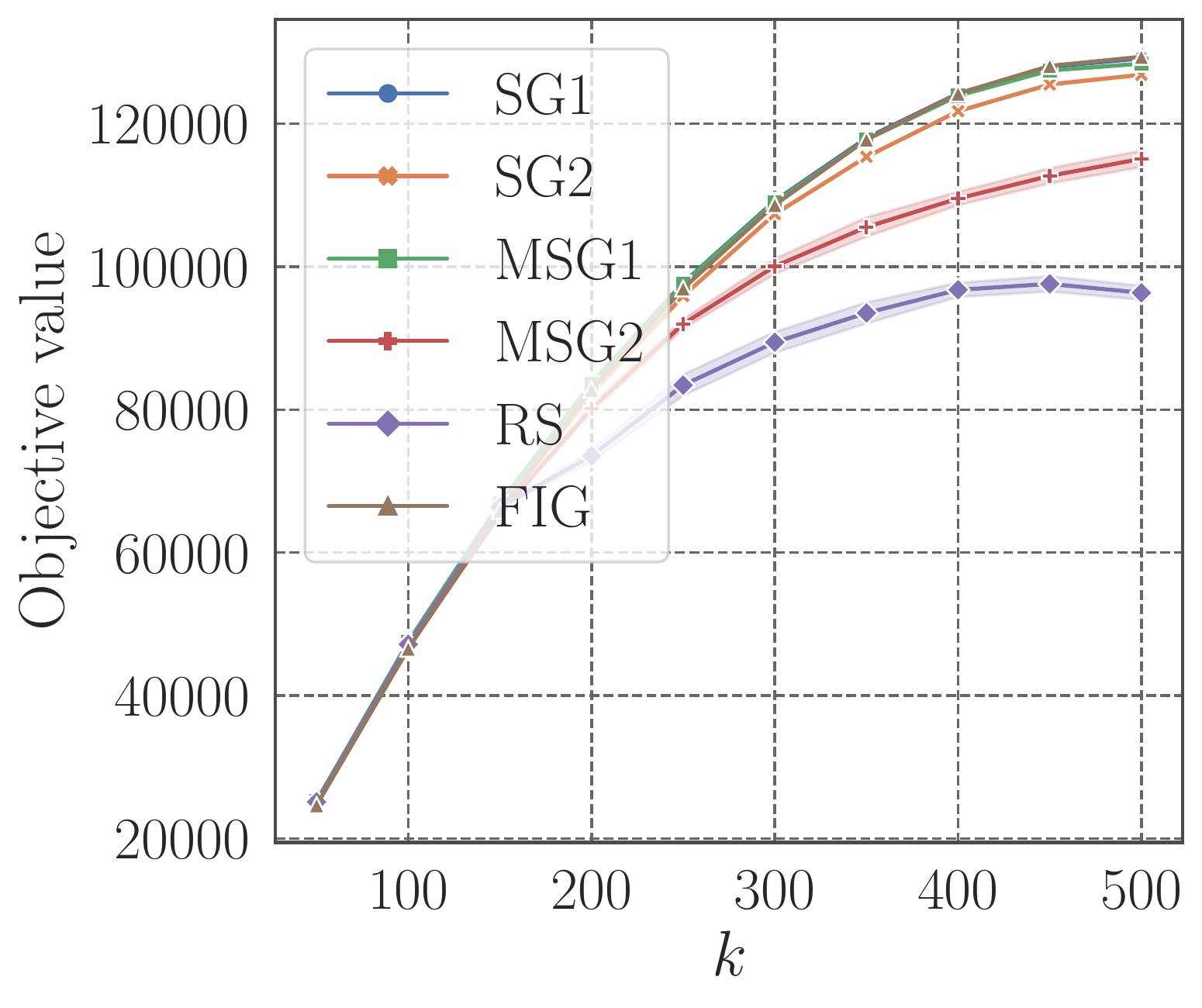}
		\subcaption{ER, Objective Value}
		\label{fig:er_v}
	\end{minipage}
	\begin{minipage}[t]{0.32\textwidth}
		\includegraphics[width=.88\textwidth]{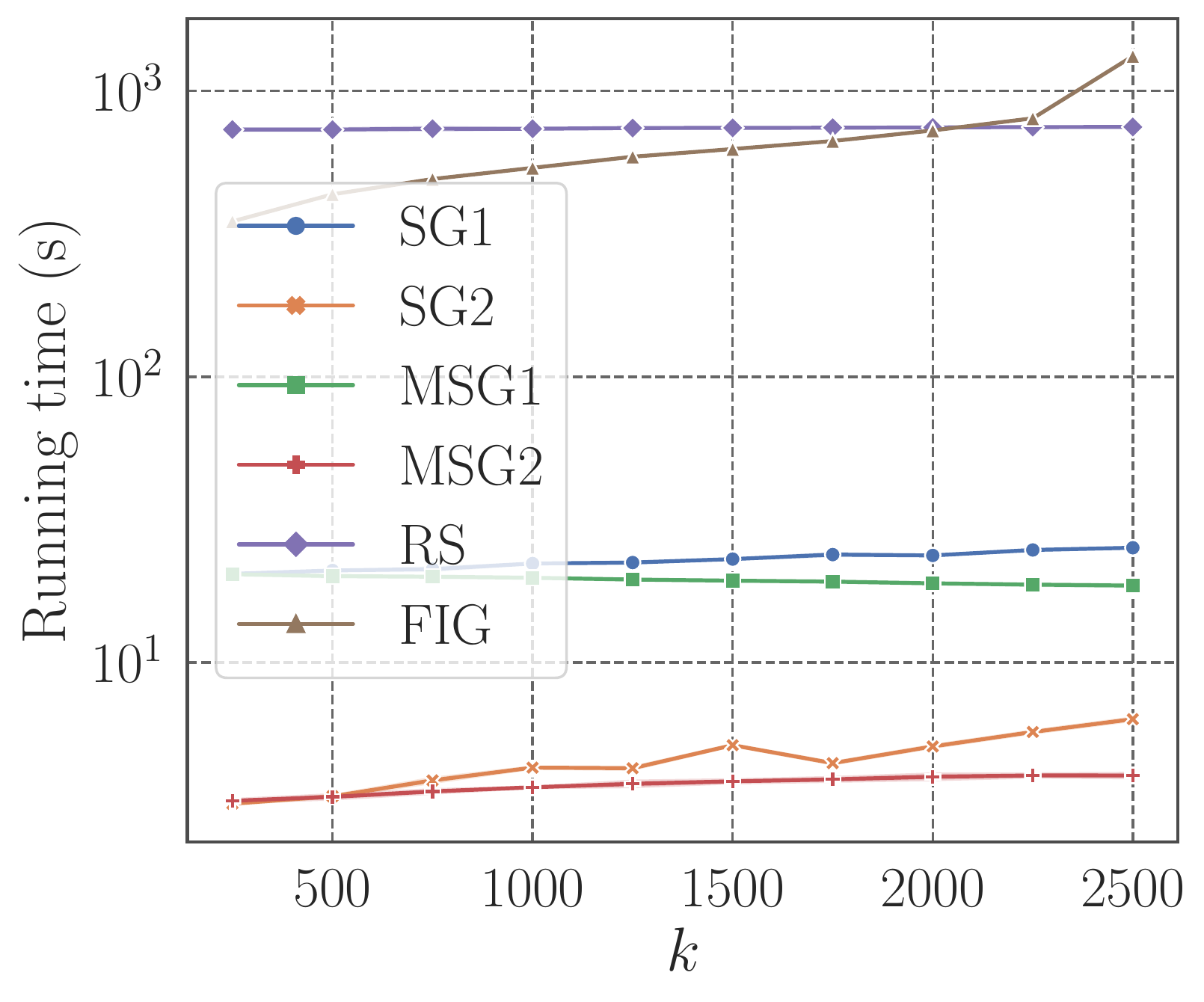}
		\subcaption{BA, Running Time (semi-log)}
		\label{fig:ba_t}
	\end{minipage}
	\begin{minipage}[t]{0.32\textwidth}
		\includegraphics[width=.88\textwidth]{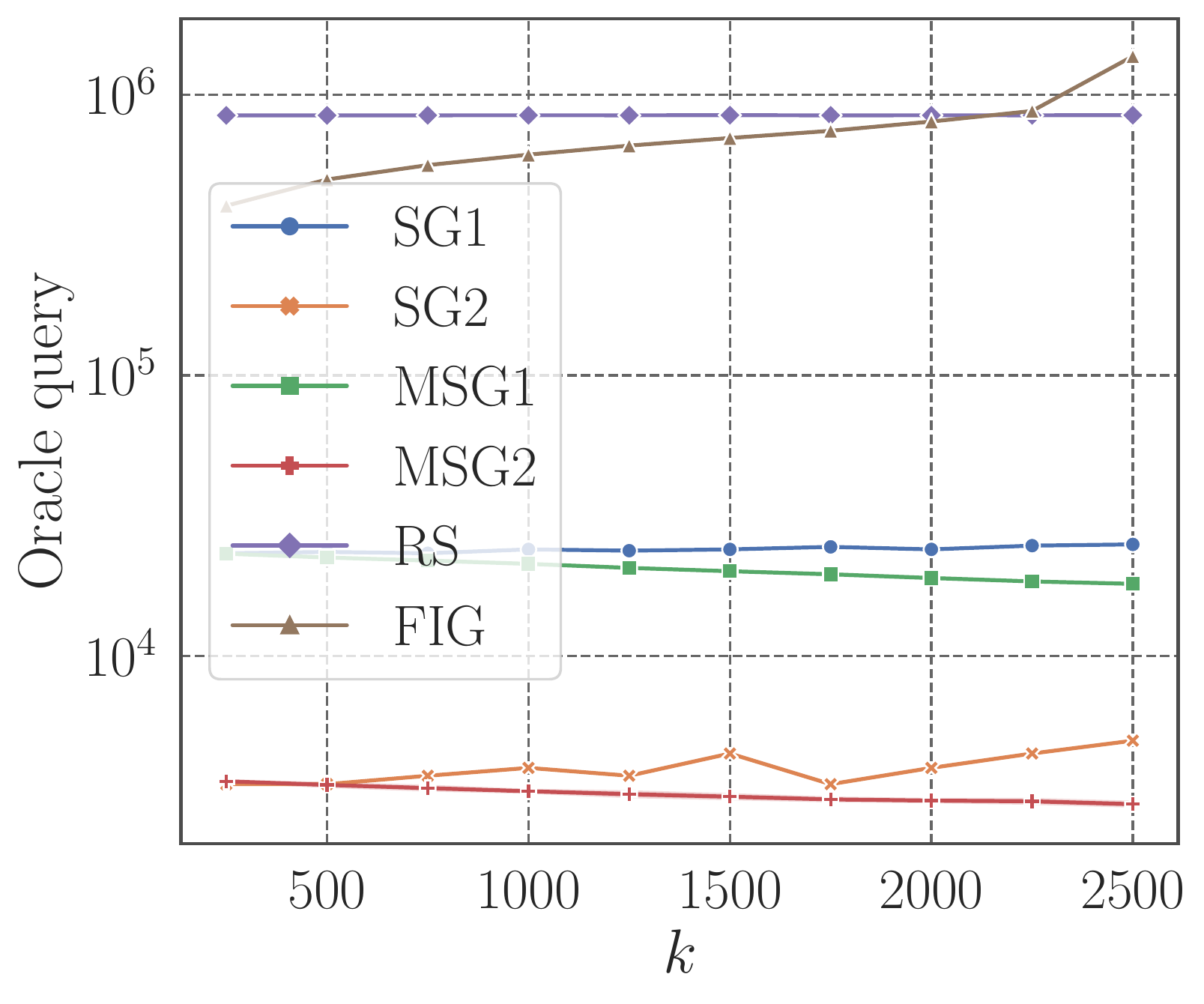}
		\subcaption{BA, Oracle Query (semi-log)}
		\label{fig:ba_o}
	\end{minipage}
	\begin{minipage}[t]{0.32\textwidth}
		\includegraphics[width=.88\textwidth]{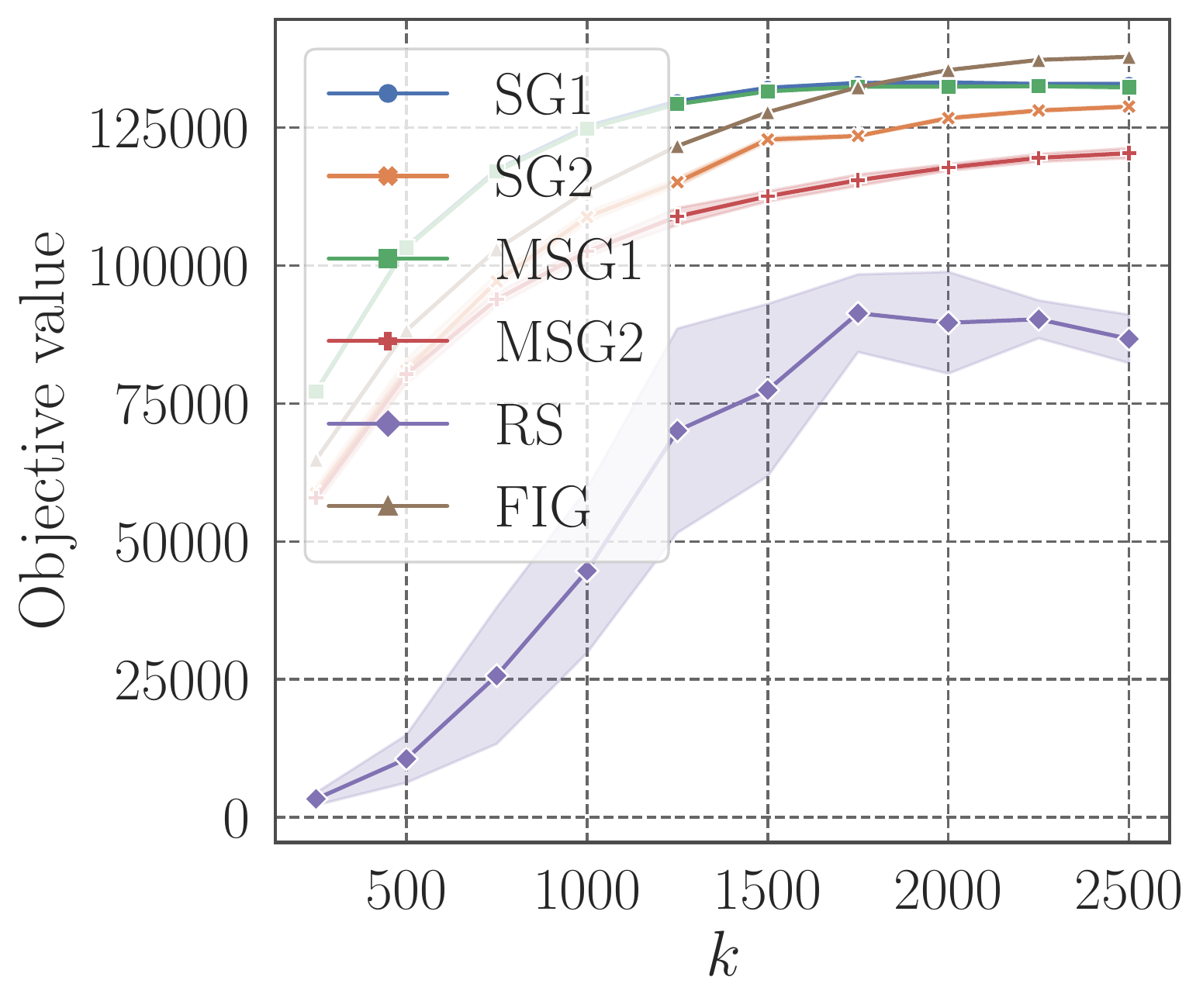}
		\subcaption{BA, Objective Value}
		\label{fig:ba_v}
	\end{minipage}
	\caption{Comparison of Algorithms with Synthetic Cut-function Maximization Instances.}
	\label{fig:syn}
\end{figure*}

We evaluate (modified) \sg\ via experiments. 
All the algorithms are implemented in Python3, 
and all the experiments are conducted
on 
a 64-bit macOS (Mojave) machine 
with 3.3 GHz Intel Core i7 CPUs and 16 GB RAM.
In \Cref{subsec:delta}, we examine 
the empirical effect of the $\delta$ value on the behavior of modified \sg.  
We then compare 
the following four kinds of algorithms with 
synthetic and real-world instances in \Cref{subsec:syn,subsec:real}, respectively.

\begin{itemize}
	\item \sg\ (\Cref{algorithm}):  
	We consider two algorithms, 
	\sg1 and \sg2, that 
	employ $\epsilon=0.01$ and $\epsilon=1/2$, respectively. 
	The approximation guarantee of \sg1 is not proved since it violates $\epsilon\ge1/\e$; 
	we here use it as a heuristic method and study its empirical behavior.  
	\sg2 achieves a 
	$\frac{1}{4}\left( 1 - 4\cdot\frac{\k-1}{\n-\k} \right)
	$-approximation guarantee if \Cref{assump:hek} holds. 
	
	\item Modified \sg\ (\msg) (\Cref{algorithm:delta}): 
	As with \sg, we consider two algorithms:  
	\msg1 ($\epsilon=0.01$) and \msg2 ($\epsilon=1/2$). 
	In \Cref{subsec:syn,subsec:real}, we let $\delta=0.1$. 
	The approximation guarantee of 
	\msg1 is not proved, while  		 
	\msg2 achieves a $0.2$-approximation guarantee. 
	
	\item Random sampling (RS) \citep{buchbinder2017comparing}: 
	A randomized $(1/\e-\epsilon)$-approximation algorithm 
	with $\Op{\frac{\n}{\epsilon^2}\ln\frac{1}{\epsilon}}$ oracle queries. 
	We set $\epsilon = 0.3$ as in the  experiments of \citep{kuhnle2019interlaced}, 
	which yields about a $0.07$-approximation guarantee.  
	
	\item Fast interlace greedy (FIG) \citep{kuhnle2019interlaced}: 
	A deterministic ($1/4-\epsilon$)-approximation algorithm with $\Op{\frac{\n}{\epsilon}\ln\frac{\n}{\epsilon}}$ oracle queries. 
	As in the experiments of \citep{kuhnle2019interlaced},  
	we set a parameter of the algorithm 
	(denoted by $\delta$ in the original paper) at $0.1$, 
	which yields a $0.1$-approximation guarantee.  
\end{itemize}

When implementing those algorithms in practice, 
we may employ various acceleration methods 
including the lazy evaluation \citep{minoux1978accelerated,leskovec2007cost} 
and memoization \citep{iyer2019memoization}.  
However, 
we here do not use them since our aim is to 
make simple and clear comparisons of the algorithms.

\subsection{Empirical Effects of $\delta$ Values}\label{subsec:delta}
We consider a synthetic instance of maximizing a cut function, 
which is non-negative and submodular. 
We construct an Erd\H{o}s--R{\'e}nyi (ER) random graph with 
$\n=100$ nodes, edge probability $p=1/2$, 
and uniform edge weights. 
The objective function to be maximized is a cut function defined on the graph, 
where we can choose up to $\k=10$ nodes.  
We apply \msg1 and \msg2 with $\delta=10^{-10},10^{-9}\dots,10^{-1}$ 
to the instance.

The numbers of oracle queries and objective values 
are shown in \Cref{fig:delta_o,fig:delta_v}, respectively, 
where each curve and error band indicate the average and standard deviation
calculated over $100$ trials.  
While the $\epsilon$ value affects the performance 
(smaller $\epsilon$ leads to better objective values and more queries), 
the increase in the $\delta$ value has little effect. 
This suggests that, 
while $\delta$ is introduced to handle the $2\cdot\frac{\k-1}{\n-\k}$ term 
that appears in the approximation ratio derived in \Cref{subsec:sg}, 
it is actually not essential. 
We leave it an open problem whether we can prove an approximation guarantee 
without introducing parameters like $\delta$. 
In what follows, we let $\delta=0.1$.

\subsection{Synthetic Instance}\label{subsec:syn}
We compare the algorithms with two synthetic cut-function maximization instances. 
One is a larger version of the above instance: 
We construct an ER random graph with $\n=1000$, $p=1/2$, and uniform edge 
weights. 
Another is defined with a Barab\'asi--Albert (BA) random graph 
with $\n=5000$ nodes and uniform edge weights, 
which is constructed as follows: 
Starting from $50$ nodes, 
we alternately add a new node and connect it to $50$ existing nodes.  
For the ER and BA instances, 
we consider various cardinality constraints with 
$\k=50,100,\dots,500$ and 
$\k=250, 500, \dots, 2500$, 
respectively. 
We apply \sg1, \sg2, \msg1, \msg2, RS, and FIG to the instances 
and observe the 
running times, numbers of oracle queries, and objective values. 
The results of the randomized algorithms are shown by the mean and standard deviation 
calculated over $10$ trails. 

\Cref{fig:syn} summarizes the results. 
With both ER and BA instances, 
\sg\ and \msg\ 
run much faster and 
require far fewer oracle queries than RS and FIG.  
For each $\epsilon$ value, 
\msg\ tends to be more efficient than \sg. 
Regarding the ER instances, 
\sg1, \msg1, and FIG achieve almost the same objective values. 
The objective value of \sg2 is slightly worse than them. 
\msg2 performs worse than those above, 
but it still outperforms RS by a considerable margin.  
As regards the BA instances, 
\sg1 and \msg1 outperform FIG when $\k$ is small, 
and the opposite is true when $\k$ is large. 
Objective values of \sg2 and \msg2 are worse than that of FIG, 
but they are far better than that of RS. 
To conclude, \sg-style algorithms are far more efficient 
than the existing methods, while achieving comparable objective values.

\subsection{Real-world  Instance}\label{subsec:real}
We compare the algorithms with real-world instances. 
We employ the mutual information as an objective function. 
Given a positive semidefinite matrix $\Xb\in\R^{V\times V}$, 
we let $\Xb[S]$ denote the principal submatrix of $\Xb$ indexed by $S\subseteq V$. 
We define the entropy function as $H(S)\coloneqq\ln\det\Xb[S]$ 
($H(\emptyset)\coloneqq0$), 
which is submodular due to the Ky Fan's inequality. 
We assume 
that the smallest eigenvalue of $\Xb$ is larger than or equal to $1$, 
which makes the entropy function monotone and non-negative. 
The mutual information is defined as  
$\f(S) = H(S) + H(V\bs S) - H(V)$, 
which is known to be submodular \citep{krause2008near,sharma2015entropy}. 
The function is non-negative since 
$\f(S) = H(S) + H(V\bs S) - H(V) \ge H(\emptyset) = 0$ for any $S\subseteq V$ 
due to the 
submodularity of $H(\cdot)$ and $H(\emptyset)=0$.

We consider a feature selection instance 
based on mutual information maximization~\citep{iyer2012algorithms,sharma2015entropy}. 
Given a matrix $\Ab$, 
whose column indices correspond to features, 
we define the mutual information with $\Xb \coloneqq \Ib + \Ab^\top \Ab$.  
To obtain matrix $\Ab$, we use 
``Geographical Original of Music'' dataset available at \citep{olson2017pmlb}. 
The dataset has $117$ features, and we create additional $\binom{117}{2}$ 
second-order polynomial feature vectors as in \citep{bertsimas2016best}. 
By adding some of them to the original 117 features and 
normalizing the columns of resulting $\Ab$, 
we obtain $\n\times \n$ matrices $\Xb$ for $\n=200, 300, \dots, 1000$. 
We let $\k=200$.  
We apply the algorithms to the instances with various $\n$ values. 
The results are again shown by the mean and standard deviation over $10$ trials. 

\Cref{fig:mi} summarizes the results. 
As with the results of synthetic instances, 
the \sg-style algorithms are far more efficient than FIG and RS. 
Oracle queries of all the algorithms increase very slowly with $\n$ in the semi-log plot, 
which is consistent with the fact that 
their oracle complexities are (nearly) linear in $\n$. 
The results of objective values are also similar to those of the synthetic instances: 
The objective values of \sg-style algorithms are as good as or slightly worse 
than that of FIG, but they are far better than that of RS. 

\begin{figure*}[tb]
	\centering
	\begin{minipage}[t]{0.32\textwidth}
		\includegraphics[width=.88\textwidth]{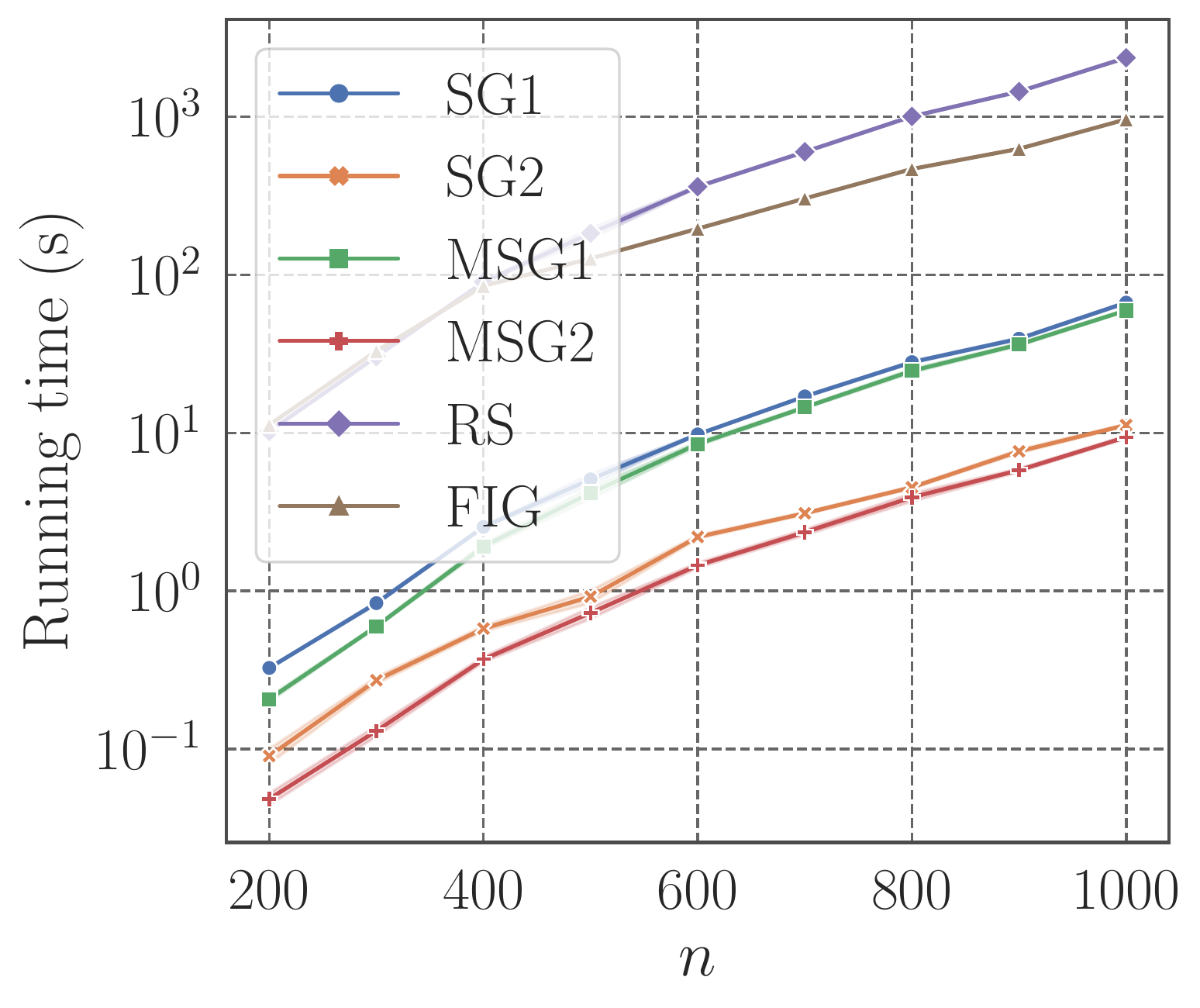}
		\subcaption{Running Time (semi-log)}
		\label{fig:mi_t}
	\end{minipage}
	\begin{minipage}[t]{0.32\textwidth}
		\includegraphics[width=.88\textwidth]{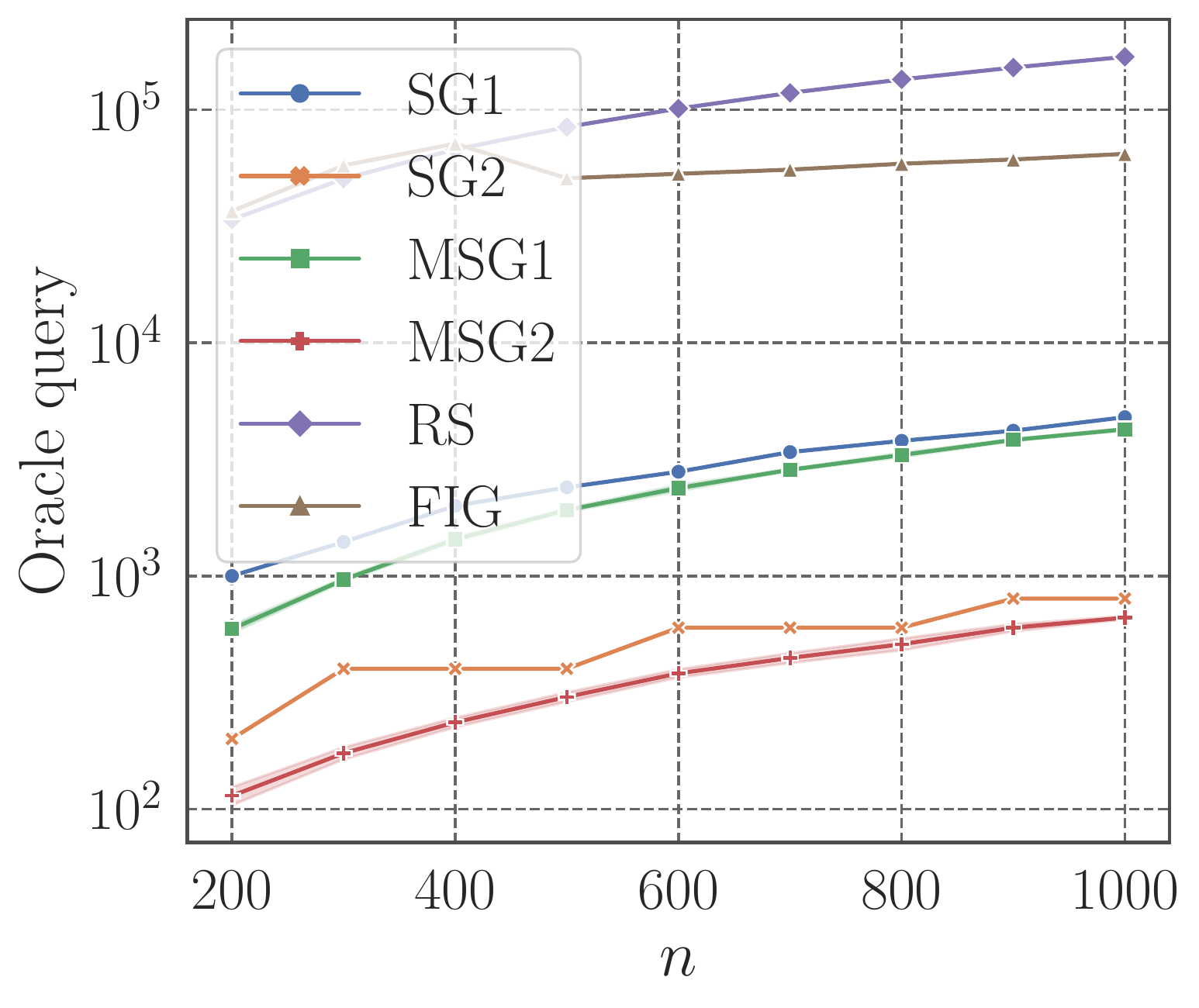}
		\subcaption{Oracle Query (semi-log)}
		\label{fig:mi_o}
	\end{minipage}
	\begin{minipage}[t]{0.32\textwidth}
		\includegraphics[width=.88\textwidth]{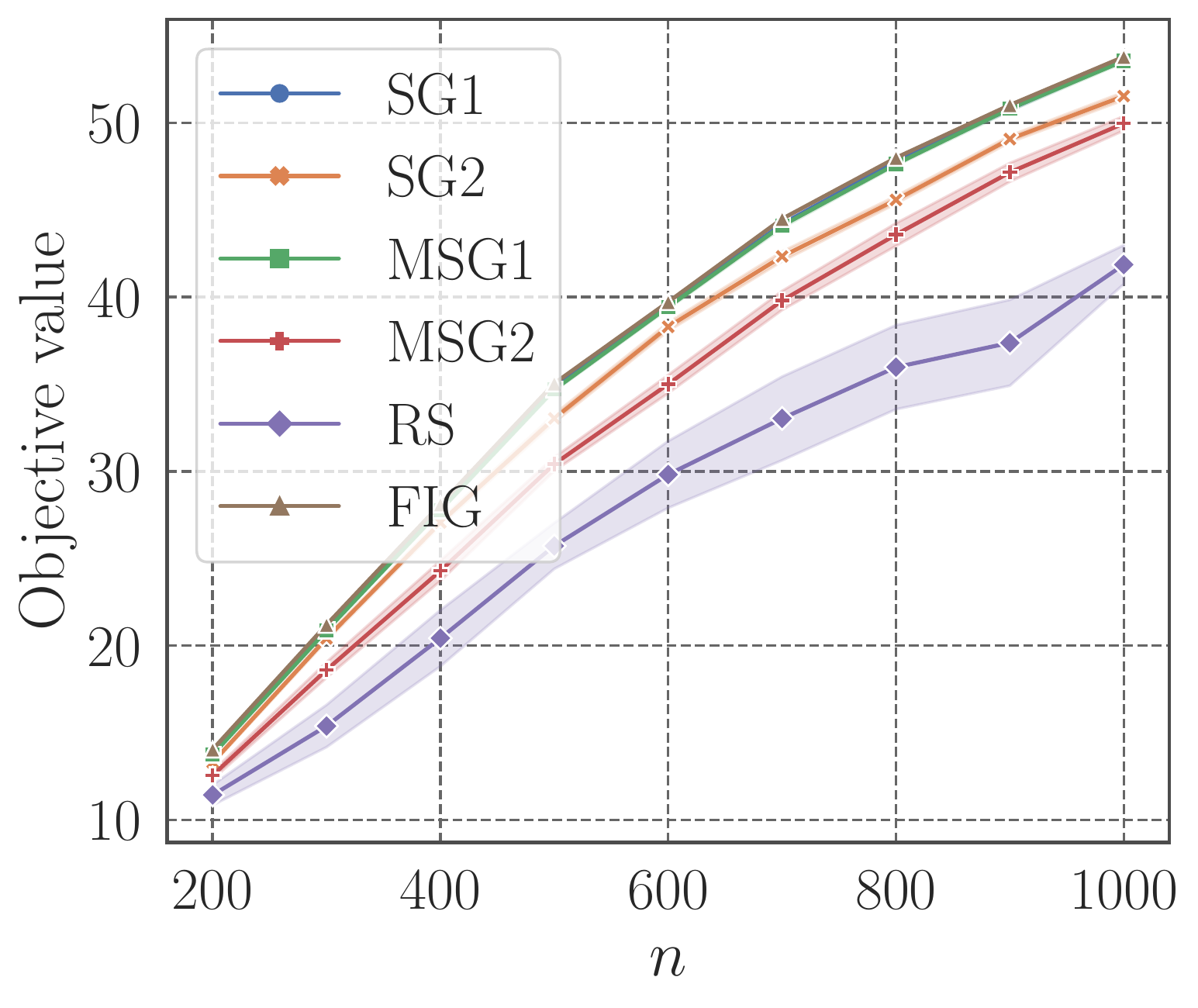}
		\subcaption{Objective Value}
		\label{fig:mi_v}
	\end{minipage}
	\caption{Comparison of Algorithms with Real-world Mutual Information Maximization Instances.}
	\label{fig:mi}
\end{figure*}	

\section{CONCLUSION AND DISCUSSION}\label{sec:conclusion}
We proved approximation guarantees of (modified) \sg\ for non-monotone submodular maximization with a cardinality constraint. 
We first proved a $\frac{1}{4}\left( 1 - 2\cdot\frac{\k-1}{\n-\k} \right)^2$-approximation guarantee of \sg\ under some assumptions;  
this yields a positive approximation ratio if $\n$ is sufficiently larger than $\k$.  
We then developed modified \sg\ and proved its $\frac{1}{4}(1-\delta)^2$-approximation guarantee 
without using the assumptions.  
We also showed that modified \sg\ 
requires at most $\n\ln2 + \n\delta\frac{\k}{\k-1}$ and 
$\max\{\n, \k + \frac{2\k}{\delta}\}\times \ln2 + \k$
oracle queries in expectation 
and in the worst-case, respectively.  
This result provides a constant-factor approximation algorithm with the fewest oracle queries. 
Experiments demonstrated that (modified) \sg\ can 
run much faster and require far fewer oracle queries than existing methods 
while achieving comparable objective values. 


\subsubsection*{Acknowledgements}
The author is grateful to 
Kaito Fujii, 
Takanori Maehara, and anonymous reviewers for providing valuable comments.

{
	\bibliography{mybib}
	\bibliographystyle{apalike}
}

\appendix
\setcounter{equation}{0}
\setcounter{thm}{0}
\setcounter{algorithm}{0}
\renewcommand{\theequation}{A\arabic{equation}}

\theoremstyle{theorem}
\newtheorem*{lem_1}{Lemma 1}
\newtheorem*{lem_4}{Lemma 4}

\clearpage

\begin{center}
	{\fontsize{18pt}{0pt}\selectfont \bf Appendix}
\end{center}

\section{PROOF OF LEMMA~\ref{lem:marginal}}\label{a_sec:lemma1}
We here prove the following lemma: 

\marginal*

\begin{proof}
	Assume that 
	all random quantities are conditioned on $\As_{i-1}$.  
	Since $\Rs$ consists of $\ceil{\si}$ elements 
	sampled uniformly at random from $V\bs\As_{i-1}$, 
	we have
	\begin{align}
	\Pr[\Rs\cap\{\Aso\bs\As_{i-1}\} = \emptyset]
	&\le
	\left(1 - \frac{|\Aso\bs\As_{i-1}|}{|V\bs \As_{i-1}|} \right)^{\ceil{\si}}
	\\&\le
	\left(1 - \frac{|\Aso\bs\As_{i-1}|}{|V\bs \As_{i-1}|} \right)^{\si}
	\\&\le
	\exp\left( -\si \frac{|\Aso\bs\As_{i-1}|}{|V\bs \As_{i-1}|}  \right)
	\\&\le
	\exp\left( -\si \frac{|\Aso\bs\As_{i-1}|}{\n}  \right)
	\end{align}
	From the concavity of 
	$1-\exp\left(-\si\frac{x}{\n}\right)$ 
	as a function of $x={|\Aso\bs\As_{i-1}|}\in[0,\k]$, 
	we obtain
	\begin{align}
	\begin{aligned}
	\Pr[\Rs\cap\{\Aso\bs\As_{i-1}\} \neq \emptyset]
	\ge{}
	\left(1 - \exp\left( -\si \frac{\k}{\n}  \right)\right) \frac{|\Aso\bs\As_{i-1}|}{\k}
	\ge{}
	(1 - \epsilon)\frac{|\Aso\bs\As_{i-1}|}{\k}
	\end{aligned}
	\label{a_eq:pr}
	\end{align}
	We now consider bounding $\f(\As_{i}) - \f(\As_{i-1})$ from below. 
	Since $\a_i$ is chosen by the greedy rule from $\Rs$ and 
	$\f(\As_{i}) - \f(\As_{i-1})$ is non-negative, 
	if $\Rs\cap\{\Aso\bs\As_{i-1}\}$ is nonempty, 
	$\f(\As_{i}) - \f(\As_{i-1})$ 
	is at least as large as 
	$\clip{\fdel{\a}{\As_{i-1}}}$ in expectation, where $\a\in V$ is chosen uniformly at random from $\Rs\cap\{\Aso\bs\As_{i-1}\}$ and 
	$\clip{x}\coloneqq\max\{x,0\}$ for any $x\in\R$. 
	Furthermore, since $\Rs$ contains each element of $\Aso\bs\As_{i-1}$ 
	equally likely, 
	we can take $\a\in V$ to be sampled uniformly at random from 
	$\Aso\bs\As_{i-1}$. As a result, we obtain 
	\begin{align}
	\E[\f(\As_{i}) - \f(\As_{i-1})]
	\ge{}&
	\Pr[\Rs\cap\{\Aso\bs\As_{i-1}\} \neq \emptyset]
	\times 
	\frac{\sum_{\a\in\Aso\bs\As_{i-1}} \clip{\fdel{\a}{\As_{i-1}}}}{|\Aso\bs\As_{i-1}|}
	\\
	\ge{}&
	\frac{1-\epsilon}{\k} \sum_{\a\in\Aso\bs\As_{i-1}}\clip{\fdel{\a}{\As_{i-1}}}
	\\
	\ge{}&
	\frac{1-\epsilon}{\k} \sum_{\a\in\Aso\bs\As_{i-1}}{\fdel{\a}{\As_{i-1}}}
	\\
	\ge{}&
	\frac{1-\epsilon}{\k}\fdel{\Aso}{\As_{i-1}}, 
	\end{align}
	where 
	the second inequality comes from \eqref{a_eq:pr} 
	and 
	the last inequality comes from the submodularity. 
	By taking expectation over all possible realizations of $\As_{i-1}$, we obtain the lemma.  
\end{proof}
Note that, the use of clipping $\clip{\cdot}$ is 
crucial when $\f$ is non-monotone. 
Without it, 
the values in the second and third lines 
can be negative due to the lack of the monotonicity; 
in this case, the inequality does not hold. 
However, thanks to the 
non-negativity of $\f(\As_{i}) - \f(\As_{i-1})$, 
we can clip the marginal gain, 
which enables us to prove the lemma even if $\f$ is non-monotone.

\section{PROOF OF LEMMA~\ref{lem:eps}} \label{a_sec:lemma4}
We here prove the following lemma: 
\eps*

To prove this, we use the following two lemmas. 

\begin{restatable}{lem}{lin}\label{lem:lin}
	If $0\le y \le x \le 1$, we have
	\[
	(x-y)^m \ge x^m - my
	\]
	for any integer $m\ge1$. 
\end{restatable}

\begin{proof}
	We prove the claim by induction. 
	If $m=1$, the inequality holds trivially. 
	Assume that it holds for every $m^\prime=1,\dots,m-1$. 
	Then, we have
	\begin{align}
	(x-y)^m 
	&\ge (x-y)(x^{m-1} - (m-1)y)
	\\
	&\ge
	x^m - (m-1)xy - x^{m-1}y
	\\
	&\ge
	x^m - my,
	\end{align}
	where the last inequality comes from $x\le1$. 
\end{proof}

\begin{lem}\label{lem:gamma}
	If $0\le\gamma\le1$, we have
	\[
	\left( 1 - \frac{\gamma}{x}  \right)^{x-1} \ge \e^{-\gamma} 
	\]
	for any $x\ge1$. 
\end{lem}

\begin{proof}
	Let $g(x)=\left(1 - \frac{\gamma}{x}\right)^{x-1}$. 
	By considering logarithmic differential, we obtain 
	\begin{align}
	\frac{\rm d}{{\rm d}x}g(x)
	&= 
	g(x)
	\left( 
	\ln\left( 1 - \frac{\gamma}{x}   \right)  
	+
	\frac{\gamma}{x} \cdot
	\frac{x-1}{x-\gamma}
	\right)
	\\
	&\le 
	g(x)
	\left( 
	\left( 1 - \frac{\gamma}{x}   \right)  - 1
	+
	\frac{\gamma}{x} \cdot
	\frac{x-1}{x-\gamma}
	\right)
	\\
	&= 
	-g(x)\frac{\gamma}{x}
	\cdot
	\frac{1 - \gamma}{x - \gamma}
	\\
	&\le 0. 
	\end{align}
	Hence $g(x)$ decreases as $x$ becomes larger. 
	Since $\lim_{x\to+\infty}g(x)=\e^{-\gamma}$, we obtain the claim. 
\end{proof} 
We are now ready to prove \Cref{lem:eps}. 
\begin{proof}[Proof of \Cref{lem:eps}]
	Note that we have 
	\[
	0 \le \frac{2}{\n-\k} \le 1 - \frac{1}{\k}\ln\frac{1}{\epsilon} \le 1
	\]
	thanks to $\k\ge2$, $\n\ge3\k$, and $\epsilon\ge1/\e$ (see, \Cref{assump:hek}). 
	Therefore, by using \Cref{lem:lin}, we obtain 
	\begin{align}
	\left( 1-\frac{1}{\k}\ln\frac{1}{\epsilon} - \frac{2}{\n-\k}\right)^{\k-1}
	\ge{}
	\left( 1-\frac{1}{\k}\ln\frac{1}{\epsilon} \right)^{\k-1}
	- 
	2\cdot\frac{\k-1}{\n-\k}.
	\end{align}
	Furthermore, since $\ln\frac{1}{\epsilon}\le1$ due to $\epsilon\ge1/\e$, 
	\Cref{lem:gamma} implies 
	\[
	\left( 1-\frac{1}{\k}\ln\frac{1}{\epsilon} \right)^{\k-1}
	\ge
	\e^{-\ln\frac{1}{\epsilon}}=\epsilon.
	\]
	Hence we obtain the claim. 
\end{proof}

\end{document}

%% file: package.tex
\usepackage{amsmath,amsthm,amssymb}

\usepackage[utf8]{inputenc} 
\usepackage[T1]{fontenc}    
\usepackage{amsfonts}       
\usepackage{nicefrac}       
\usepackage{microtype}      

\usepackage{graphicx,color}
\usepackage{subcaption}
\usepackage{wrapfig}

\usepackage{tabularx}       
\usepackage{booktabs}       

\usepackage{algorithm} 
\usepackage{algorithmicx}
\usepackage[noend]{algpseudocode}

\usepackage{thmtools}
\usepackage{thm-restate}
\usepackage{url}            
\usepackage{cleveref}
\usepackage{autonum}
\usepackage{mathtools}
\usepackage{tikz}
\usetikzlibrary{positioning, arrows.meta, decorations.pathmorphing}

\usepackage{braket}
\usepackage{multirow}
\usepackage{lipsum}
\allowdisplaybreaks
\usepackage{lscape}

%% file: macro.tex
\newcommand{\Abar}{\overline{A}}

\newcommand{\Vbar}{\overline{V}}

\renewcommand{\hbar}{\overline{h}}

\newcommand{\sbar}{\overline{s}}

\newcommand{\R}{\mathbb{R}}
\newcommand{\N}{\mathbb{N}}

\newcommand{\E}{\mathbb{E}}

\newtheorem{assump}{Assumption}
\newtheorem{thm}{Theorem}
\newtheorem{lem}{Lemma}

\theoremstyle{definition}

\DeclareMathOperator*{\maximize}{maximize}
\DeclareMathOperator{\subto}{subject\ to}
\DeclareMathOperator*{\argmax}{argmax}

\newcommand{\ceil}[1]{\left\lceil {#1} \right\rceil}

\newif\iffigure
\figuretrue

\mathcode`@="8000 
{\catcode`\@=\active\gdef@{\mkern1mu}}

\mathchardef\Gamma="7100 \mathchardef\Delta="7101
\mathchardef\Theta="7102 \mathchardef\Lambda="7103
\mathchardef\Psi="7104 \mathchardef\Pi="7105 \mathchardef\Sigma="7106
\mathchardef\Upsilon="7107 \mathchardef\Phi="7108
\mathchardef\Psi="7109 \mathchardef\Omega="710A




%% file: main.bbl
\begin{thebibliography}{}

\bibitem[Badanidiyuru and Vondr\'{a}k, 2014]{badanidiyuru2014fast}
Badanidiyuru, A. and Vondr\'{a}k, J. (2014).
\newblock Fast algorithms for maximizing submodular functions.
\newblock In {\em Proceedings of the 25th Annual ACM-SIAM Symposium on Discrete
  Algorithms}, pages 1497--1514. SIAM.

\bibitem[Balkanski et~al., 2018]{balkanski2018nonmonotone}
Balkanski, E., Breuer, A., and Singer, Y. (2018).
\newblock Non-monotone submodular maximization in exponentially fewer
  iterations.
\newblock In {\em Advances in Neural Information Processing Systems 31}, pages
  2353--2364. Curran Associates, Inc.

\bibitem[Bertsimas et~al., 2016]{bertsimas2016best}
Bertsimas, D., King, A., and Mazumder, R. (2016).
\newblock Best subset selection via a modern optimization lens.
\newblock {\em Ann. Statist.}, 44(2):813--852.

\bibitem[Buchbinder and Feldman, 2018]{buchbinder2018deterministic}
Buchbinder, N. and Feldman, M. (2018).
\newblock Deterministic algorithms for submodular maximization problems.
\newblock {\em ACM Trans. Algorithms}, 14(3):32:1--32:20.

\bibitem[Buchbinder et~al., 2014]{buchbinder2014submodular}
Buchbinder, N., Feldman, M., Naor, J.~S., and Schwartz, R. (2014).
\newblock Submodular maximization with cardinality constraints.
\newblock In {\em Proceedings of the 25th Annual ACM-SIAM Symposium on Discrete
  Algorithms}, pages 1433--1452. SIAM.

\bibitem[Buchbinder et~al., 2017]{buchbinder2017comparing}
Buchbinder, N., Feldman, M., and Schwartz, R. (2017).
\newblock Comparing apples and oranges: Query trade-off in submodular
  maximization.
\newblock {\em Math. Oper. Res.}, 42(2):308--329.

\bibitem[de~Veciana et~al., 2019]{veciana2019stochasticgreedy}
de~Veciana, G., Hashemi, A., and Vikalo, H. (2019).
\newblock Stochastic-greedy++: Closing the optimality gap in exact weak
  submodular maximization.
\newblock {\em arXiv preprint arXiv:1907.09064}.

\bibitem[Ene et~al., 2019]{ene2019submodular}
Ene, A., Nguyen, H., and Vladu, A. (2019).
\newblock Submodular maximization with matroid and packing constraints in
  parallel.
\newblock In {\em Proceedings of the 51st Annual ACM SIGACT Symposium on Theory
  of Computing}, pages 90--101. ACM.

\bibitem[Fahrbach et~al., 2019]{fahrbach2019non-monotone}
Fahrbach, M., Mirrokni, V., and Zadimoghaddam, M. (2019).
\newblock Non-monotone submodular maximization with nearly optimal adaptivity
  and query complexity.
\newblock In {\em Proceedings of the 36th International Conference on Machine
  Learning}, volume~97, pages 1833--1842. PMLR.

\bibitem[Feldman et~al., 2017]{feldman2017greed}
Feldman, M., Harshaw, C., and Karbasi, A. (2017).
\newblock Greed is good: Near-optimal submodular maximization via greedy
  optimization.
\newblock In {\em Proceedings of the 2017 Conference on Learning Theory},
  volume~65, pages 758--784. PMLR.

\bibitem[{Feldman} et~al., 2011]{feldman2011unified}
{Feldman}, M., {Naor}, J., and {Schwartz}, R. (2011).
\newblock A unified continuous greedy algorithm for submodular maximization.
\newblock In {\em Proceedings of the 2011 IEEE 52nd Annual Symposium on
  Foundations of Computer Science}, pages 570--579.

\bibitem[Gharan and Vondr\'{a}k, 2011]{gharan2011submodular}
Gharan, S.~O. and Vondr\'{a}k, J. (2011).
\newblock Submodular maximization by simulated annealing.
\newblock In {\em Proceedings of the 22nd Annual ACM-SIAM Symposium on Discrete
  Algorithms}, pages 1098--1116. SIAM.

\bibitem[Gupta et~al., 2010]{gupta2010constrained}
Gupta, A., Roth, A., Schoenebeck, G., and Talwar, K. (2010).
\newblock Constrained non-monotone submodular maximization: Offline and
  secretary algorithms.
\newblock In {\em Proceedings of the 6th International Conference on Internet
  and Network Economics}, pages 246--257. Springer-Verlag.

\bibitem[Harshaw et~al., 2019]{harshaw2019submodular}
Harshaw, C., Feldman, M., Ward, J., and Karbasi, A. (2019).
\newblock Submodular maximization beyond non-negativity: Guarantees, fast
  algorithms, and applications.
\newblock In {\em Proceedings of the 36th International Conference on Machine
  Learning}, volume~97, pages 2634--2643. PMLR.

\bibitem[{Hashemi} et~al., 2018]{hasemi2018randomized}
{Hashemi}, A., {Ghasemi}, M., {Vikalo}, H., and {Topcu}, U. (2018).
\newblock A randomized greedy algorithm for near-optimal sensor scheduling in
  large-scale sensor networks.
\newblock In {\em 2018 Annual American Control Conference}, pages 1027--1032.

\bibitem[Hassidim and Singer, 2017]{hassidim2017robust}
Hassidim, A. and Singer, Y. (2017).
\newblock Robust guarantees of stochastic greedy algorithms.
\newblock In {\em Proceedings of the 34th International Conference on Machine
  Learning}, volume~70, pages 1424--1432. PMLR.

\bibitem[Iyer and Bilmes, 2012]{iyer2012algorithms}
Iyer, R. and Bilmes, J. (2012).
\newblock Algorithms for approximate minimization of the difference between
  submodular functions, with applications.
\newblock In {\em Proceedings of the 28th Conference on Uncertainty in
  Artificial Intelligence}, pages 407--417. AUAI Press.

\bibitem[Iyer and Bilmes, 2019]{iyer2019memoization}
Iyer, R. and Bilmes, J. (2019).
\newblock A memoization framework for scaling submodular optimization to large
  scale problems.
\newblock In {\em Proceedings of the 22nd International Conference on
  Artificial Intelligence and Statistics}, volume~89, pages 2340--2349. PMLR.

\bibitem[Ji et~al., 2020]{ji2020stochastic}
Ji, S., Xu, D., Li, M., Wang, Y., and Zhang, D. (2020).
\newblock Stochastic greedy algorithm is still good: Maximizing submodular +
  supermodular functions.
\newblock In {\em Optimization of Complex Systems: Theory, Models, Algorithms
  and Applications}, pages 488--497. Springer International Publishing.

\bibitem[Khanna et~al., 2017]{khanna2017approximation}
Khanna, R., Elenberg, E.~R., Dimakis, A.~G., Ghosh, J., and Negahban, S.
  (2017).
\newblock On approximation guarantees for greedy low rank optimization.
\newblock In {\em Proceedings of the 34th International Conference on Machine
  Learning}, volume~70, pages 1837--1846. PMLR.

\bibitem[Krause et~al., 2008]{krause2008near}
Krause, A., Singh, A., and Guestrin, C. (2008).
\newblock Near-optimal sensor placements in gaussian processes: Theory,
  efficient algorithms and empirical studies.
\newblock {\em J. Mach. Learn. Res.}, 9(Feb):235--284.

\bibitem[Kuhnle, 2019]{kuhnle2019interlaced}
Kuhnle, A. (2019).
\newblock Interlaced greedy algorithm for maximization of submodular functions
  in nearly linear time.
\newblock In {\em Advances in Neural Information Processing Systems 32}, pages
  2371--2381. Curran Associates, Inc.

\bibitem[Lee et~al., 2010]{lee2010maximizing}
Lee, J., Mirrokni, V., Nagarajan, V., and Sviridenko, M. (2010).
\newblock Maximizing nonmonotone submodular functions under matroid or knapsack
  constraints.
\newblock {\em SIAM J. Discrete. Math.}, 23(4):2053--2078.

\bibitem[Leskovec et~al., 2007]{leskovec2007cost}
Leskovec, J., Krause, A., Guestrin, C., Faloutsos, C., VanBriesen, J., and
  Glance, N. (2007).
\newblock Cost-effective outbreak detection in networks.
\newblock In {\em Proceedings of the 13th ACM SIGKDD International Conference
  on Knowledge Discovery and Data Mining}, pages 420--429. ACM.

\bibitem[Lin and Bilmes, 2010]{lin2010multi}
Lin, H. and Bilmes, J. (2010).
\newblock Multi-document summarization via budgeted maximization of submodular
  functions.
\newblock In {\em Proceedings of Human Language Technologies: The 2010 Annual
  Conference of the North American Chapter of the Association for Computational
  Linguistics}, pages 912--920. Association for Computational Linguistics.

\bibitem[Minoux, 1978]{minoux1978accelerated}
Minoux, M. (1978).
\newblock Accelerated greedy algorithms for maximizing submodular set
  functions.
\newblock In {\em Proceedings of the 8th IFIP Conference on Optimization
  Techniques}, pages 234--243. Springer.

\bibitem[Mirzasoleiman et~al., 2016]{mirzasoleiman2016fast}
Mirzasoleiman, B., Badanidiyuru, A., and Karbasi, A. (2016).
\newblock Fast constrained submodular maximization: Personalized data
  summarization.
\newblock In {\em Proceedings of The 33rd International Conference on Machine
  Learning}, volume~48, pages 1358--1367. PMLR.

\bibitem[Mirzasoleiman et~al., 2015]{mirzasoleiman2015lazier}
Mirzasoleiman, B., Badanidiyuru, A., Karbasi, A., Vondr\'{a}k, J., and Krause,
  A. (2015).
\newblock Lazier than lazy greedy.
\newblock In {\em Proceedings of the 29th AAAI Conference on Artificial
  Intelligence}, pages 1812--1818. AAAI Press.

\bibitem[Nemhauser et~al., 1978]{nemhauser1978analysis}
Nemhauser, G.~L., Wolsey, L.~A., and Fisher, M.~L. (1978).
\newblock An analysis of approximations for maximizing submodular set
  functions-{I}.
\newblock {\em Math. Program.}, 14(1):265--294.

\bibitem[Olson et~al., 2017]{olson2017pmlb}
Olson, R.~S., La~Cava, W., Orzechowski, P., Urbanowicz, R.~J., and Moore, J.~H.
  (2017).
\newblock {PMLB}: {A} large benchmark suite for machine learning evaluation and
  comparison.
\newblock {\em BioData Min.}, 10(1):36.

\bibitem[Pan et~al., 2014]{pan2014parallel}
Pan, X., Jegelka, S., Gonzalez, J.~E., Bradley, J.~K., and Jordan, M.~I.
  (2014).
\newblock Parallel double greedy submodular maximization.
\newblock In {\em Advances in Neural Information Processing Systems 27}, pages
  118--126. Curran Associates, Inc.

\bibitem[Qian et~al., 2018]{qian2018approximation}
Qian, C., Yu, Y., and Tang, K. (2018).
\newblock Approximation guarantees of stochastic greedy algorithms for subset
  selection.
\newblock In {\em Proceedings of the 27th International Joint Conference on
  Artificial Intelligence}, pages 1478--1484. International Joint Conferences
  on Artificial Intelligence Organization.

\bibitem[Sharma et~al., 2015]{sharma2015entropy}
Sharma, D., Kapoor, A., and Deshpande, A. (2015).
\newblock On greedy maximization of entropy.
\newblock In {\em Proceedings of the 32nd International Conference on Machine
  Learning}, pages 1330--1338. PMLR.

\bibitem[{Song} et~al., 2017]{song2017deep}
{Song}, H.~O., {Jegelka}, S., {Rathod}, V., and {Murphy}, K. (2017).
\newblock Deep metric learning via facility location.
\newblock In {\em Proceedings of the 2017 IEEE Conference on Computer Vision
  and Pattern Recognition}, pages 2206--2214.

\bibitem[Vondr{\'a}k, 2013]{vondrak2013symmetry}
Vondr{\'a}k, J. (2013).
\newblock Symmetry and approximability of submodular maximization problems.
\newblock {\em SIAM J. Comput.}, 42(1):265--304.

\bibitem[Wei et~al., 2014]{wei2014fast}
Wei, K., Iyer, R., and Bilmes, J. (2014).
\newblock Fast multi-stage submodular maximization.
\newblock In {\em Proceedings of the 31st International Conference on Machine
  Learning}, volume~32, pages 1494--1502. PMLR.

\end{thebibliography}
